%% file: main.tex
\documentclass[lettersize,journal]{IEEEtran}
\usepackage{amsmath,amssymb,amsfonts,amsthm, amssymb,mathrsfs}
\usepackage{mathtools}
\usepackage{array}
\usepackage{stfloats}
\usepackage{graphicx}
\usepackage{hyperref}
\usepackage{url}
\usepackage{xr}
\usepackage{verbatim}
\usepackage{cite}
\usepackage{textcomp}
\usepackage{tabularx}
\usepackage[table]{xcolor} 
\usepackage{colortbl}
\usepackage{subcaption}
\usepackage{algorithmic}
\usepackage{mathrsfs}
\usepackage[ruled]{algorithm}
\newtheorem{theorem}{Theorem}
\newtheorem{proposition}{Proposition}
\newtheorem{lemma}{Lemma}

\newcommand\norm[1]{\left\lVert#1\right\rVert}

\newcommand{\Rd}{\mathbb{R}^d}

\newcommand{\Ast}[1]{\boldsymbol{A}_{\text{coor},#1}}
\newcommand{\bst}[1]{\boldsymbol{b}_{\text{coor},#1}}
\newcommand{\x}{\boldsymbol{\mathrm{x}}}
\newcommand{\X}{\boldsymbol{x}}

\newcommand{\Hm}{\frac{H_\epsilon}{M}  }
\usepackage{tcolorbox}
\begin{document}
\title{Service Placement in Small Cell Networks Using Distributed Best Arm Identification in Linear Bandits}
\author{  Mariam Yahya, Aydin Sezgin,  \textit{Senior Member, IEEE},  and Setareh Maghsudi\\
\thanks{This research was supported by the Federal Ministry of Research, Technology and Space (BMFTR) under Grants 16KIS2411 and the German Research Foundation (DFG) under Grant MA7111/6-1.

{Parts of this paper were presented at the European Signal Processing Conference (EUSIPCO 2025) \cite{yahyaEUSIPCO}. }

Mariam Yahya is with the Department of Computer Science, University of Tübingen, 72076 Tübingen, Germany  and with the Faculty of Electrical Engineering and Information Technology, Ruhr-University Bochum, 44801 Bochum, Germany (e-mail: mariam.yahya@uni-tuebingen.de).

Aydin Sezgin is with the Faculty of Electrical Engineering and Information Technology, Ruhr-University Bochum, 44801 Bochum, Germany 
(e-mail: aydin.sezgin@rub.de).

Setareh Maghsudi is with the Faculty of Electrical Engineering and Information Technology, Ruhr-University Bochum, 44801 Bochum, Germany 
(e-mail: setareh.maghsudi@rub.de).
}
}
\maketitle
\begin{abstract}
As users in small cell networks increasingly rely on computation-intensive services, cloud-based access often results in high latency. Multi-access edge computing (MEC) mitigates this by bringing computational resources closer to end users, with small base stations (SBSs) serving as edge servers to enable low-latency service delivery. However, limited edge capacity makes it challenging to decide which services to deploy locally versus in the cloud, especially under unknown service demand and dynamic network conditions.
To tackle this problem, we model service demand as a linear function of service attributes and formulate the service placement task as a linear bandit problem, where SBSs act as agents and services as arms. The goal is to identify the service that, when placed at the edge, offers the greatest reduction in total user delay compared to cloud deployment.
We propose a distributed and adaptive multi-agent best-arm identification (BAI) algorithm under a fixed-confidence setting, where SBSs collaborate to accelerate learning. Simulations show that our algorithm identifies the optimal service with the desired confidence and achieves near-optimal speedup, as the number of learning rounds decreases proportionally with the number of SBSs. We also provide theoretical analysis of the algorithm’s sample complexity and communication overhead. 
\end{abstract}


\begin{table*}[b]
\centering
\begin{tcolorbox}[colback=gray!8,colframe=red,width=0.95\textwidth]
\footnotesize
Accepted for publication in IEEE Transactions on Mobile Computing. © 2026 IEEE. Personal use of this material is permitted. Permission
from IEEE must be obtained for all other uses, in any current or future media, including reprinting/republishing this material for advertising or
promotional purposes, creating new collective works, for resale or redistribution to servers or lists, or reuse of any copyrighted component of this
work in other works.
\end{tcolorbox}
\end{table*}
\section{Introduction}
The rapid advancements in mobile communication technologies have profoundly transformed our society and reshaped sectors such as entertainment \cite{khan2022survey}, healthcare \cite{liu2024energy}, and transportation \cite{li2024joint}. This transformation is accelerating with the widespread demand for 5G networks, which is expected to reach 6.3 billion users by 2030 \cite{ericsson2024}, thereby increasing the need for innovative and scalable network infrastructures. One promising technology is mobile cloud computing, which allows mobile devices to offload computationally intensive tasks to a cloud equipped with large computation and storage capabilities, supporting complex applications. However, the network cloud is often placed too far from the end users, creating high network delays and increased traffic. To mitigate these issues, multi-access edge computing (MEC) has emerged as a promising paradigm that places computation and storage resources at the network edge, closer to end users \cite{liang2017mobile}. 
While MEC can significantly reduce latency and improve quality of service, it also introduces new challenges in resource allocation and service placement due to the limited edge resources, dynamic user behavior, and randomness of the environment.

The service placement problem addresses the challenge of determining which services to place at the resource-limited network edges. The services are typically hosted at the cloud, with only a subset placed at the edges, closer to the end users. When a user requests a service, the edge server responds if it hosts it. Otherwise, the task is offloaded to the cloud or other edge servers \cite{malazi2022dynamic,ouyang2019adaptive}. The benefits of effective service placement include reducing user-perceived delay, minimizing backhaul traffic, lowering service costs \cite{malazi2022dynamic}, or meeting budget constraints \cite{chen2019budget}. All these benefits depend largely on the demand (or popularity) of the services, which is not known a priori, making the problem of service placement quite challenging. 

Some studies make the simplifying assumption that service demand, or its average value, is either known or can be directly estimated from recent observations \cite{chen2025multi,farhadi2021service,yang2019cloudlet}. However, this assumption does not accurately capture real-world service demand dynamics. To address this challenge, multi-armed bandit (MAB) approaches have been employed to learn service demand in an online manner \cite{chen2019budget, chen2018spatio, he2020bandit}. MAB is a sequential decision-making framework in which, in each round, the learner selects one of $K$ arms (options) and receives a reward drawn from an unknown but fixed distribution. The most common MAB objective is to maximize cumulative reward or, equivalently, minimize cumulative regret \cite{lattimore2020bandit}. Another key objective is to identify the arm that maximizes expected reward, a problem known as best arm identification (BAI), which falls under the pure exploration paradigm \cite{bubeck2009pure}.

The contextual MAB (CMAB) framework has been applied to service placement as it leverages contextual information, such as service attributes or network conditions  to enhance learning efficiency \cite{chen2019budget, ouyang2019adaptive,chen2018spatio}. Most existing works using standard MABs \cite{wang2023online} or CMABs \cite{chen2019budget,ouyang2019adaptive } in service placement focus on minimizing cumulative regret, which requires balancing exploration (gathering information about new services) and exploitation (selecting the best-known service). While this approach is effective in dynamic settings where cost is measured by the obtained reward \cite{ouyang2019adaptive,wang2023online}, it is less suitable for service placement problems, where the goal is to identify and deploy the best service for extended periods, often lasting hours \cite{chen2019budget}. In such cases, the BAI framework is more appropriate, as it allows for an initial exploration phase to identify the best service, followed by long-term deployment.

In this paper, we consider a small cell network consisting of a number of small base stations (SBSs) with computational resources, representing the edge servers in a MEC architecture. The network also includes a macro base station (MBS) that provides wide-area coverage for all users, forwards user service requests to the cloud, and facilitates information exchange among SBSs.
We propose a distributed and adaptive multi-agent BAI algorithm in linear bandits that enables the SBSs to collaborate in identifying the best service placement with the desired confidence. 
Our approach reduces the number of learning rounds proportionally to the number of SBSs.

\subsection{Paper Contributions}
\label{sec:contributions}
In this paper, we propose a distributed and adaptive multi-agent BAI algorithm in linear bandits. We then formulate the optimal service placement problem at the SBSs as a collaborative multi-agent BAI problem. Due to resource limitations at the SBSs, they can host only one service, while the remaining services are hosted in the cloud. Our objective is to identify the service that achieves the highest reduction in the total users' delay when placed at the SBSs instead of the cloud.

Since the delays and demand for different services is unknown, the SBSs collaborate to learn the delay improvement sequentially by leveraging contextual information that reflects the underlying structure of the bandit problem. Each SBS adapts its decisions according to past observations, enabling more informed actions in subsequent rounds. To accelerate this learning process, SBSs share data derived from their selected services and the corresponding observed rewards via the MBS. Because all SBSs access the same set of services and have the same bandit model, this collaboration significantly improves the overall learning efficiency. However, exchanging information across SBSs at every round can lead to excessive communication costs. To mitigate this, SBSs exchange information only when there is a significant change in their locally observed data, rather than in every round.

In summary, the contributions of this paper are as follows:  
\begin{itemize}
    \item We propose a distributed and fully adaptive algorithm for multi-agent best arm identification in the linear bandit setting under a fixed-confidence framework.  
    \item We formulate the service placement problem in small cell networks as a contextual linear bandit problem, aiming to identify the service that, when placed at the SBSs, maximizes the reduction in the total user-perceived delay compared to cloud deployment.
    \item To the best of our knowledge, this is the first work to apply the BAI setting in linear bandits to a MEC problem, including service placement. We employ our proposed distributed BAI algorithm to determine the optimal service to deploy at the SBSs, leveraging collaboration among them to minimizes the number of learning rounds required by each SBS.  
    \item We establish theoretical guarantees on the sample complexity and the number of communication rounds of the proposed algorithm. 
    \item We analyze the robustness of the proposed algorithm to communication failures and show how it can be extended to heterogeneous agents.
    \item Numerical results demonstrate the effectiveness of our proposed algorithm in identifying the best arm (service) on both synthetic data and in the service placement problem. The algorithm identifies the best arm with the specified confidence while accelerating the learning process by nearly $M$-fold compared to independent learning by individual SBSs, where $M$ denotes the number of SBSs (agents).
\end{itemize}
\subsection{Paper Organization}
The remainder of this chapter is organized as follows. Section~\ref{sec:related_work} presents the related work. Section~\ref{sec:system_model} introduces the system model for service placement in small-cell networks, while Section~\ref{sec:problem_formulation} formally defines the problem. Section~\ref{sec:modeling} formulates the service placement problem as a distributed BAI problem and demonstrates how the proposed algorithm addresses it. In Section~\ref{sec:DistLinGapE}, we present our distributed BAI algorithm within the linear bandit framework. Section~\ref{sec:extension} shows how the proposed work can be extended or applied to 
heterogeneous settings and the limitations, while~Section~\ref{sec:robustness} analyzes the algorithm's robustness to communication failures. Section~\ref{sec:numerical_results} provides numerical results, and Section~\ref{sec:conclusion} concludes the chapter and discusses future research directions.

\section{Related Work}
\label{sec:related_work}

\subsection{BAI in Linear Bandits}
Linear bandits (LB) are a variant of MABs in which the learner observes a $d$-dimensional context vector for each selected arm and observes a reward that is a noisy linear function of the context \cite{abbasi2011improved}. In this setting, the learner aims to estimate the unknown parameter vector that defines the linear relationship between the context and the expected rewards. The shared linear structure across arms enables the agent to learn more efficiently compared to standard unstructured MABs, where the mean rewards of the arms are estimated independently. The standard stochastic MAB is a special case of LB when the set of available actions in each round is the standard orthonormal basis of $\Rd$.  

In BAI, the learner sequentially samples arms with the objective of ultimately recommending a single best arm, that is, the one with the highest expected reward \cite{bubeck2009pure}. BAI can be studied under two primary settings: fixed-budget and fixed-confidence. In the fixed-budget setting, the goal is to minimize the probability of misidentifying the best arm within a predefined number of rounds \cite{audibert2010best}. In the fixed-confidence setting, the objective is to identify the best arm with a specified level of confidence while minimizing the number of rounds required. In this setting, identification strategies are characterized by three fundamental components: the arm selection strategy, the stopping condition, and the arm recommendation rule \cite{kaufmann2016complexity}.

The strategies for BAI in LBs differ from those for conventional MABs because of the correlation between the arms. In LBs, pulling a suboptimal arm can still improve the estimation of the underlying linear model. The first algorithm for BAI in linear bandits appeared in \cite{hoffman2014correlation}, but it is designed for the fixed budget case and does not take  the complexity of the linear structure into account. Soare et al. \cite{soare2014best} were the first to address the BAI problem in LB under the fixed-confidence setting. They proposed the $\mathcal{XY}$-Static algorithm, which uses optimal experimental design for arm selection but suffers from high sample complexity due to its inability to adapt to the outcomes of the previous rounds. To enhance efficiency, they introduced the $\mathcal{XY}$-Adaptive strategy, which divides the total number of rounds into multiple phases. In each phase, a static allocation algorithm is used, but the strategy leverages the outcomes of previous phases to identify important directions and eliminate less promising ones. This approach can be considered semi-adaptive, as it does not adjust arm selection based on the results of each round. 
In contrast, Xu et al. \cite{xu2018fully} introduced a fully adaptive algorithm, named Linear Gap-based Exploration (LinGapE), that dynamically adjusts the arm selection at each round based on all past observations, thus improving the sample complexity.

The standard MAB problem is extended to the distributed case where agents collaboratively maximize the cumulative regret \cite{wang2019distributed, shi2021federated,chen2023federated,he2022simple} or identify the best arm for the pure exploration case \cite{mitra2021exploiting,tao2019collaborative}. The agents can access the same set of arms \cite{shi2021federated, mitra2021exploiting,huang2021federated} or different subsets of arms  \cite{chen2023federated}, and in most cases they exchange information through a central coordinator. In addition to the objective of minimizing the cumulative regret or minimizing the sample complexity, distributed algorithms aim to minimize the communication cost whether it is measured in terms of the number of communication rounds \cite{mitra2021exploiting} or the amount of real numbers or bits transmitted \cite{wang2019distributed, amani2023distributed,huang2021federated}. Meeting both objectives is challenging because there is a tradeoff between the number of local updates and the number of communication rounds. In elimination-based algorithms, for example, reducing local rounds can result in the elimination of the optimal arm, while reducing communication can increase the sample complexity. 

In the literature, studies on distributed BAI in linear bandits for the fixed-confidence setting are still limited. In this paper, we propose a distributed and adaptive multi-agent BAI algorithm in LB for the fixed confidence setting. The proposed algorithm is named Distributed Linear Gap-based Exploration (DistLinGapE) as it mainly extends the work in \cite{xu2018fully} to the distributed setting. We also derive the sample complexity for the proposed algorithm and the upper bound on the number of communication rounds.
\subsection{The Service Placement Problem}
MEC problems are uncertain due to the randomness of the environment, user behavior, and the availability of network resources. Consequently, MAB-based strategies have been extensively applied in this domain \cite{chen2019budget, ouyang2019adaptive, chen2018spatio, yahya2024decentralized,zhang2017contextual}. In particular, the service placement problem in MEC is challenging due to the unknown service demand. To address this, several works have modeled it as a contextual bandit problem \cite{chen2019budget, ouyang2019adaptive, chen2018spatio}, with the modeling approach and solution strategy varying depending on the placement objectives and constraints. In \cite{chen2019budget}, a service provider rents computational resources at edge servers to host application services, so it uses a CMAB approach to find the best placement that improves the users' QoS within the budget requirements. CMAB is also used in \cite{chen2018spatio} to find the optimal placement of a subset of services under a budget constraint first for non-overlapping SBSs, and then for overlapping coverage areas, adding difficulty to the problem as the service can be requested from multiple SBSs.

The work in \cite{ouyang2019adaptive} considers the service placement problem from the users' perspective with the objective of jointly optimizing the user’s perceived latency and service migration cost, weighted by user preferences. The problem is modeled as a CMAB problem where the context vector describes the users' state information such as  location and service demand, and is solved using a Thompson-sampling-based algorithm. Contextual bandits are also used in content caching problems \cite{zhang2017contextual}; these problems differ from service placement, as they are concerned with accessing the cached data without processing.

Services deployed at the network edge typically remain in place for extended periods. Farhadi et al. \cite{farhadi2021service} propose a two-time-scale framework for service placement and request scheduling. In their approach, service placement decisions are made on a larger time scale, in terms of frames, while request scheduling is performed on a per-slot basis, depending on the real-time resource availability at each edge server. Their model assumes that the average demand for each service is known. This two-time-scale framework can be particularly beneficial for energy-constrained devices, such as unmanned aerial vehicles (UAVs), where minimizing service caching overhead is critical \cite{zhou2022two}. These works highlight the importance of efficient long-term placement strategies and reinforce the relevance of our contribution, which uses BAI for long-term deployment while addressing collaboration and unknown service demand.

\subsection{Algorithmic Choice for Service Placement}
Service placement decisions are made on a slower timescale than request scheduling. Consequently, the objective is to identify, with high confidence, the service to be placed locally  based on its long-term performance, which motivates a BAI formulation in linear bandits.
Specifically, we adopt the LinGapE algorithm \cite{xu2018fully} as a fully adaptive arm selection approach and extend it to a collaborative multi-agent setting. Among existing BAI methods for linear bandits \cite{jedra2020optimal,jourdan2022choosing,azizi2023meta}, some provide strong theoretical guarantees and, in some cases, improved sample complexity. However, they often incur higher computational overhead or rely on assumptions that limit their applicability in distributed service placement systems. In contrast, LinGapE combines adaptivity with moderate computational complexity, making it well suited to our problem.

For example, the Lazy Track-and-Stop algorithm \cite{jedra2020optimal} is asymptotically optimal, achieving the optimal sampling rate as the error probability $\delta \rightarrow 0$. However, for moderate values of $\delta$, LinGapE can be more sample-efficient \cite{tirinzoni2022elimination}. Moreover, Track-and-Stop requires solving a sample allocation optimization problem at each round, which results in significant computational overhead which motivated the introduction of a lazy variant. In our service placement problem, some level of error is acceptable and $\delta$ need not approach zero, especially when doing so incurs additional computational cost.

Other approaches focus on $\epsilon$-best-arm identification, to identify any arm whose performance is within $\epsilon$ of the optimal one \cite{jourdan2022choosing}. Although this relaxation can reduce sample complexity, extending $\epsilon$-BAI to a distributed multi-agent setting is considerably more challenging as agents may observe different services as locally $\epsilon$-optimal. This requires global consensus on arm elimination and coordinated sampling strategies across agents, thus increasing communication and problem complexity.

Meta-learning-based methods \cite{azizi2023meta} improve regret performance by transferring knowledge across multiple tasks under a fixed budget. This differs from our work, where agents collaboratively learn within a single bandit instance and jointly estimate the same underlying parameter vector over time. 

Overall, LinGapE provides a balance between statistical efficiency, computational complexity, and scalability. Its extension to a distributed multi-agent framework enables cooperative service placement while maintaining a good tradeoff between learning speed and communication cost, which is essential for practical  deployments.

\section{System Model}
\label{sec:system_model}
We first introduce the notation used throughout this work. 
\textbf{Notations:} We use boldface lowercase letters and boldface uppercase letters to represent vectors and matrices respectively. We use $[K]$ to represent a set of integers $\{1, \dots, K\}$. Besides, $\norm{\boldsymbol{x}}_p$ denotes the $p-$norm of a vector $\boldsymbol{x}$ and the weighted $2-$norm of vector $\boldsymbol{x}$ is defined by $\norm{\boldsymbol{x}}_{\boldsymbol{A}} = \sqrt{\boldsymbol{x}^T\boldsymbol{A}\boldsymbol{x}}$. Furthermore, we use $\mathbb{P}[\cdot]$ to denote the probability of the term inside the brackets.

We consider a system comprising an MBS connected to the network cloud and $M$ homogeneous SBSs. Each SBS serves as an edge server and provides communication coverage to users within its own non-overlapping small cell, resulting in $M$ coverage areas of similar size. The cloud is equipped with large storage and computational resources, enabling it to host and process all user service requests forwarded by the MBS. The MBS provides wide-area coverage, collects user service requests, and relays them to the cloud. In contrast, the SBSs have significantly limited computational capabilities. They are connected to the MBS in a star-shaped network topology to enable information exchange. The system model is illustrated in~Fig \ref{fig:system_model}.

A service provider manages $K$ distinct services that can be deployed either at the cloud or at the SBSs, depending on the SBSs' decisions. Each service $k \in [K]$ is characterized by a $d$-dimensional context vector $\boldsymbol{x}_k \in \Rd$. This vector can represent various factors such as the service type, computational and memory requirements, and geographic or time-based demand. 

Here, our optimization problem depends on both the service demand and delay improvement of placing the service in the SBS instead of the cloud. Therefore, we construct the context vector as
\begin{equation}
    \X_k = \tilde{G}_k \tilde{\boldsymbol{D}}_k ,
\end{equation}
where $\tilde{G}_k$ is the average delay improvement achieved when service $k$ is placed at the SBS instead of the cloud for one request, and $\tilde{\boldsymbol{D}}_k \in \Rd$ is the long-term average service demand over $d$ consecutive time periods.
The duration of the time periods is selected based on the temporal dynamics of the network, where shorter periods allows the context to adapt to highly dynamic environments, while longer periods are more suitable for networks with stable or slowly varying demand.
In our numerical results, $\tilde{\boldsymbol{D}}_k$ contains the users' average demand for service $k$ over 8 time periods, representing a four-minute interval.

Due to limited edge resources, each SBS $m \in [M]$ can host only one service at any given time. If a user within the coverage area of SBS $m$ requests a service that is unavailable at SBS $m$, it has to offload the request to the cloud via the MBS, incurring additional delay. The decision of which service to deploy at the SBSs depends on the objective of this placement. In this work, we aim to identify the service that reduces the total users' delay the most when placed at the SBSs instead of the cloud. The key challenge arises from the uncertainty in delays and service demand, which are not known in advance.

The service placement process operates in discrete time steps indexed by $t = 1, 2, \dots$. At each $t$, each SBS selects a service $a_{t,m} \in [K]$ with vector $\X_{a_{t,m}}$ and observes the resulting delay.
To facilitate learning the delay improvement and effectively model the dependence of the service demand on the contextual information, we adopt the common approach of representing service demand as a noisy linear function of the context \cite{yang2018content, zhang2017contextual}. Mathematically, the demand for service $a_{t,m}$ at time $t$ and SBS $m$ is given by  
\begin{equation}
p_{t,m} = \tilde{\boldsymbol{D}}_{a_{t,m}}^\top \boldsymbol{\omega^*} + \xi_{t,m},
\label{eq:context_demand}
\end{equation}
where $\boldsymbol{\omega}^* \in \mathbb{R}^d$ is an unknown parameter vector representing the underlying relationship between services' features and the user demands, and $\xi_{t,m}$ is the zero-mean $R_{\xi}$ sub-Gaussian noise.

The users' delay depends on whether the requested service is placed at the SBS or the cloud, as detailed below.
\begin{figure}
    \centering
    \includegraphics[width=1\linewidth]{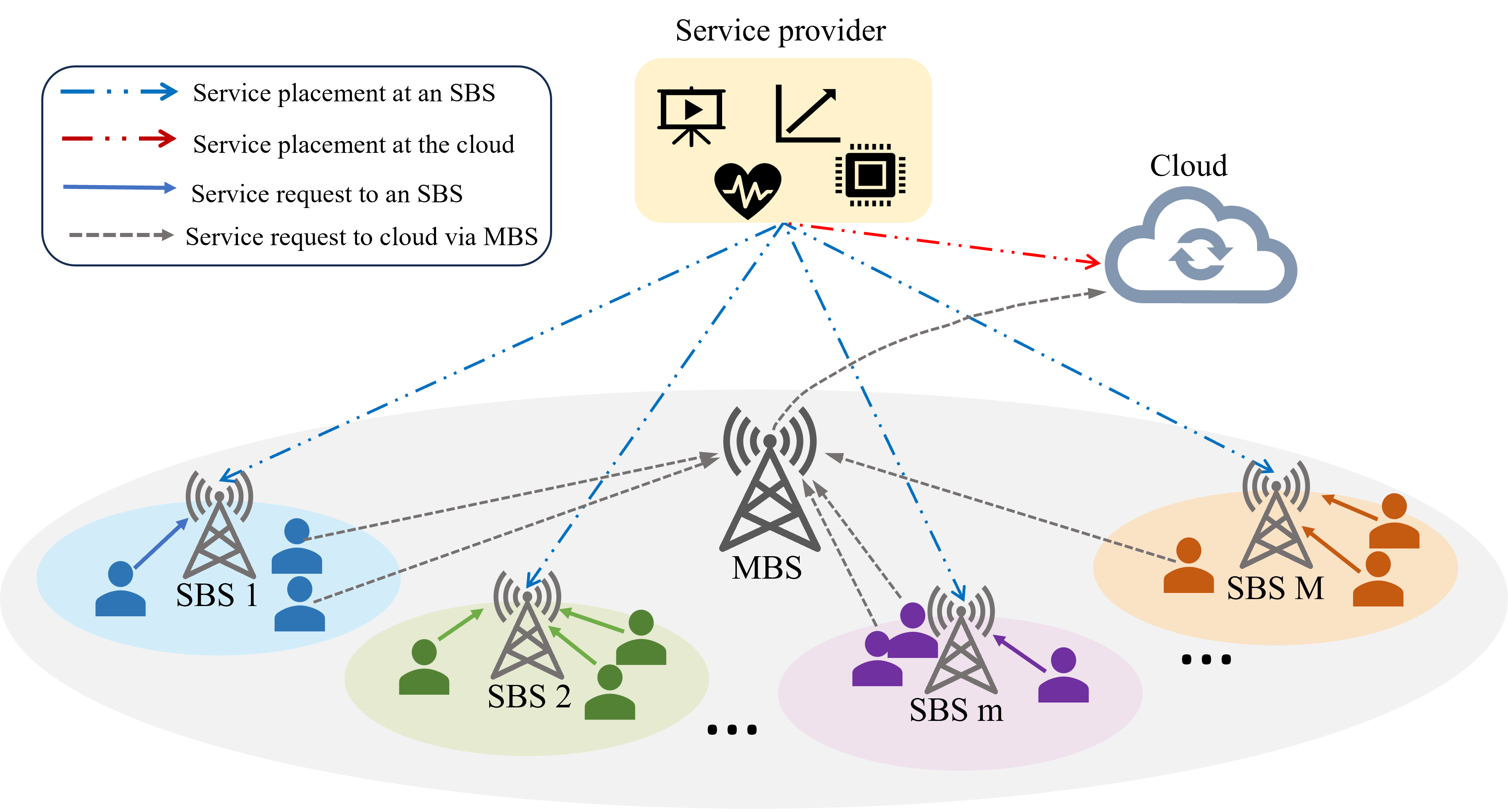}
    \caption{System model}
    \label{fig:system_model}
\end{figure}
\subsection{Service Delay at an SBS}
If a service is placed at an SBS, the service delay consists of the transmission delay of tasks from the users to the SBS and the processing delay at the SBS. We assume that the task response delay is negligible \cite{chen2018spatio}. For the uplink delay, we use the average data rate across all users because we are concerned with long-time deployment, and the instantaneous user data rates are unknown. The uplink transmission rate to an SBS, averaged across the users, is 
\begin{equation}
    \rho_{t,m} = W \log \left( 1+ \frac{P h_{t,m}}{I+\sigma_{\text{N}}^2 } \right) ,
    \label{eq:datarate}
\end{equation}
where $W$ is the uplink channel bandwidth, $P$ is the users' transmission power, $h_{t,m}$ is the channel gain at time $t$, $I$ is the interference power, and $\sigma_{\text{N}}^2$ is the noise power.  The round trip delay is $RTT_{t,m}$.
Let $s_k$ be the size of task $k$.  Then, the communication delay of task $k$ at SBS $m$ is given by $ d_{t,m}^{\text{cm}}(k)=\frac{s_k}{\rho_{t,m}}+RTT_{t,m}$ . 
The processing delay at SBS $m$ is given by $ d_{t,m}^{\text{proc}}(k) = \frac{c s_k}{f_{t,k,m}}$, where $f_{t,k,m}$ is the CPU capacity assigned to service $k$ at SBS $m$ and time $t$, and $c$ is the number of CPU cycles required per bit. Thus, the total delay for service $k$ at SBS $m$ and time $t$ is:
\begin{equation}
    d_{t,m} (k)= d_{t,m}^{\text{cm}} (k)+ d_{t,m}^{\text{proc}} (k)=  \frac{s_k}{\rho_{t,m}} +RTT_{t,m}+ \frac{c s_k}{f_{t,k,m}}.
\end{equation}
\subsection{Service Delay at the Cloud}
If the service requested by users in the service area of SBS $m$ is not hosted at the SBS, the requests are offloaded to the cloud via the MBS, incurring transmission and processing delays. Let $\rho_{t,0}$ be the average data rate between the users and the MBS, where the index $0$ is used for the cloud, indicating that the service request is relayed there. The wireless transmission delay is then given by $\frac{s_k}{\rho_{t,0}}$. Considering the delay caused by congestion and the routing policies in the backbone network between the MBS and the cloud server, the transmission delay over the backbone, with data rate $\rho_t^{\text{b}}$, is $\frac{s_k}{\rho_t^{\text{b}}}$. Let $RTT_{t,0}$ denote the round trip delay. The overall communication delay is therefore $d_{t,0}^{\text{cm}} (k)= \frac{s_k}{\rho_{t,0}} + \frac{s_k}{\rho_t^{\text{b}}}+ RTT_{t,0}$. The processing delay at the cloud is given by $d^{\text{proc}}_{t,0} (k)= \frac{c s_k}{f_{t,k,0}}$, where $f_{t,k,0}$ is the CPU capacity allocated for service $k$ at time $t$. The total service delay of task $k$ at the cloud is thus given by
\begin{equation}
    d_{t,0}(k) =  d_{t,0}^{\text{cm}} (k)+ d_{t,0}^{\text{proc}}(k).
\end{equation}
%
\subsection{Service Placement Utility}
The utility of a service placement decision is defined as the improvement in the total user delay resulting from placing the service at the SBSs instead of the cloud. Let $a_{t,m} \in [K]$ be the service placed at SBS $m$ at time $t$.
The average delay gain per request when service $k$ is placed locally is modeled by the random variable $G_{t,m}(k)$, defined as
$G_{t,m}(k)
\coloneqq
\mathbb{E}_u\!\left[d_{t,0}^u(k) - d_{t,m}^u(k)\right],$
where the expectation is taken over the users requesting service $k$.

If $p_{t,m}(a_{t,m})$ users request service $a_{t,m}$ at time $t$, the corresponding total delay improvement is
$U_{t,m}(a_{t,m})
=
G_{t,m}(a_{t,m}) \, p_{t,m}(a_{t,m}).$

Both the average delay gain per request $G_{t,m}(k)$ and the service demand $p_{t,m}(k)$ are random. Since service placement decisions operate on a slow timescale, we approximate $G_{t,m}(k)$ by its long-term average
\begin{equation}
\tilde{G}_k \coloneqq \mathbb{E}\!\left[G_{t,m}(k)\right],
\end{equation}
which is consistent with long-term service placement decisions and the use of the linear demand model. Substituting the demand model into the utility function, we obtain
\begin{equation}
    U_{t,m} (a_{t,m}) \approx \tilde{G}_{a_{t,m}} \left( \tilde{\boldsymbol{D}}_{a_{t,m}}^\top \boldsymbol{\omega}^* + \xi_{t,m}  \right) = \X_{a_{t,m}}^\top \boldsymbol{\omega}^* + \eta_{t,m}
\end{equation}
where $\boldsymbol{x}_{a_{t,m}} \coloneqq  \tilde{G}_{a_{t,m}}\tilde{\boldsymbol{D}}_{a_{t,m}}$ and $\eta_{t,m} \coloneqq  \tilde{G}_{a_{t,m}}\xi_{t,m}$ is an independent, zero-mean, $R$-sub-Gaussian noise variable with
$R =\max_{k \in [K]}| \tilde{G}_k| \, R_\xi$.

Therefore, the instantaneous delay improvement observed by SBS $m$ can be expressed as a noisy linear reward
\begin{equation}
   r_{t,m}(a_{t,m}) = \boldsymbol{x}_{a_{t,m}}^\top \boldsymbol{\theta}^{*} + \eta_{t,m}
   \label{eq:reward},
\end{equation}
where $\boldsymbol{\theta}^*$ is an unknown parameter vector used to obtain the linear bandit formulation. Since the expected total delay improvement satisfies
$
\mathbb{E}\!\left[r_{t,m}(a_{t,m}) \mid \mathcal{F}_{t-1}, a_{t,m}\right]
= \boldsymbol{x}_{a_{t,m}}^\top \boldsymbol{\theta}^{*},
$
where $\mathcal{F}_{t-1}$ denotes the filtration up to time $t-1$, the same parameter vector governs the demand model and the reward model. Hence, for notational convenience, we write $\boldsymbol{\theta}^* = \boldsymbol{\omega}^*$.

%
\section{Problem Formulation}
\label{sec:problem_formulation}
Given the limited computation and storage resources at the SBSs, they aim to identify and deploy the service that yields the greatest reduction in the total users' delay compared to cloud deployment. Since the deployed service is typically used over an extended period, making an optimal selection with high confidence is crucial. Therefore, we formulate this problem as one where the SBS must identify the service that maximizes the expected utility with high confidence while minimizing the number of decision-making rounds. Let
\begin{equation}
    a^{*}  =  \underset{a \in [K]}{\arg \max} \, \boldsymbol{x}_{a}^\top \boldsymbol{\theta}^{*}
\end{equation}
be the optimal service that maximizes the expected utility. We aim to identify the estimated best service $\hat{a}_{m}^{*}$ such that
\begin{equation}
    \mathbb{P}\left[(\boldsymbol{x}_{a^{*}}  - \boldsymbol{x}_{\hat{a}_{m}^{*}} )^{\top}
    \boldsymbol{\theta}^{*}  > \epsilon  \right] \le \delta \label{eq:stop_cond}
\end{equation}          
as fast as possible. Here, $\epsilon$ is the desired accuracy and $\delta$ is the confidence. \footnote{Since we aim to identify the best arm (service), we assume that $\epsilon =0$, but the work remains valid as an $(\epsilon, \delta)$-probably approximately correct (PAC) algorithm where $\epsilon>0$.} Since all SBSs access the same set of services and share the same vector $\boldsymbol{\theta}^{*}$, there is a single optimal service for all SBSs, i.e., $\hat{a}^{*}=\hat{a}_{m}^{*}, \forall m \in [M]$. A key challenge in this problem arises from the uncertainty in service demand, which is unknown a priori.
\section{Modeling Service Placement as a Distributed BAI Problem}
\label{sec:modeling}
In this section, we demonstrate that the service placement problem can be modeled as a distributed BAI problem within the linear bandit framework. We then solve it using the proposed DistLinGapE algorithm, presented in Section~\ref{sec:DistLinGapE}. 

Each SBS $m \in [M]$ acts as an agent in the bandit problem, while each service $k \in [K]$ corresponds to an arm associated with a context vector $\X_k \in \Rd$ characterizing the service. Selecting an arm in the bandit framework corresponds to an SBS $m$ hosting a service. Observing the reward at time $t$ involves observing the total reduction in delay resulting from this service placement among all users in the coverage area of SBS $m$. The MBS acts as a central coordinator, facilitating information exchange among the SBSs.

The best service can be identified by iteratively learning the parameter vector $\boldsymbol{\theta}^{*}$ using the BAI algorithm. Since all SBSs share the same vector $\boldsymbol{\theta}^{*}$ and set of services, and can communicate through the MBS, collaboration between the SBSs speeds up the identification process compared to independent learning by individual SBSs. 

At a high level, the process of identifying the best service placement proceeds as follows. At time $t$, each SBS $m$ hosts a service $a_{t,m}$ and observes the corresponding reward (utility). This reward is used to update the estimate of the unknown parameter vector, denoted by $\hat{\boldsymbol{\theta}}_{t,m}$, which guides the next service placement decision.  To minimize the communication overhead, each SBS continues hosting services and updating its local information until a significant change is detected at one of the SBSs. When such a change occurs, a communication round is triggered, and the SBSs share their observations via the MBS. This process repeats until the best service is identified. Section~\ref{sec:DistLinGapE} provides a detailed description of the proposed algorithm. To facilitate the explanation of the algorithm and the subsequent analysis, we will use the bandit terminology. The application to the small cell network follows simply by mapping each agent to an SBS and each arm to a service.
\section{The DistLinGapE Algorithm}
\label{sec:DistLinGapE}
We consider a distributed LB problem  with $M$ agents, where $m \in [M]$, that can exchange updates through a central coordinator. All agents have access to the same set of $K$ arms, where each arm $k \in [K]$ is associated with a $d$-dimensional context vector $\boldsymbol{x}_k \in \Rd$, with $\| \boldsymbol{x}_k \| \le L$ for some constant $L$. In round $t$, agent $m$ pulls arm $a_{t,m} \in [K]$ and receives a reward that is a noisy linear function of the context vector as given by \eqref{eq:reward}, where $\| \boldsymbol{\theta}^{*} \| \le S$. The aim of our distributed algorithm is to enable the agents to collaboratively learn $\boldsymbol{\theta}^{*}$ to identify the best arm satisfying \eqref{eq:stop_cond} by sharing observations over multiple communication rounds. Specifically, we aim to find the best arm with high confidence in the fewest number of rounds and a low communication overhead. The DistLinGapE algorithm is given in Algorithm~\ref{alg:DistLinGapE} on the next page.

Before discussing the algorithm, we give a high-level explanation of how the best arm is identified in linear bandits.
Soare et al. \cite{soare2014best} show that for each arm there exists a cone defining the set of $\boldsymbol{\theta}$ for which arm $k$ is optimal, given by  $\mathcal{C}(k) = \{\boldsymbol{\theta} \in \Rd, k \in \arg \max_{k \in [K]} \boldsymbol{x}^\top_k \boldsymbol{\theta} \}$. The optimal arm is identified when the confidence set of the estimate $\hat{\boldsymbol{\theta}}_{t,m}$ shrinks into one of the cones. The illustration in Fig.~\ref{fig:conf_cone} gives an example with $K=3$ arms and $d=2$. 
The $\mathbb{R}^2$ Euclidean space is partitioned into $3$ cones and arm $1$ is the optimal arm because $\boldsymbol{\theta}^{*} \in \mathcal{C}(1)$. In each round, the arm selection strategy aims to pull in the direction that speeds up the refinement of the confidence set (the yellow ellipsoid) to align with a cone. Fig.~\ref{fig:before_pull} shows the confidence set at time $t$. Since it overlaps with $\mathcal{C}(1)$ and $\mathcal{C}(3)$, arms 1 and 3 can be optimal. To shrink the confidence set into one cone only, the next arm pull must be in the direction $\X_3-\X_2$. Fig.~\ref{fig:before_pull} shows the confidence set after the arm pull. Since the confidence set resides fully in cone 1, arm 1 is identified as the optimal.
\begin{figure}[htbp]
    \centering
    \begin{subfigure}{0.48\columnwidth}
        \centering
        \includegraphics[width=0.8\columnwidth]{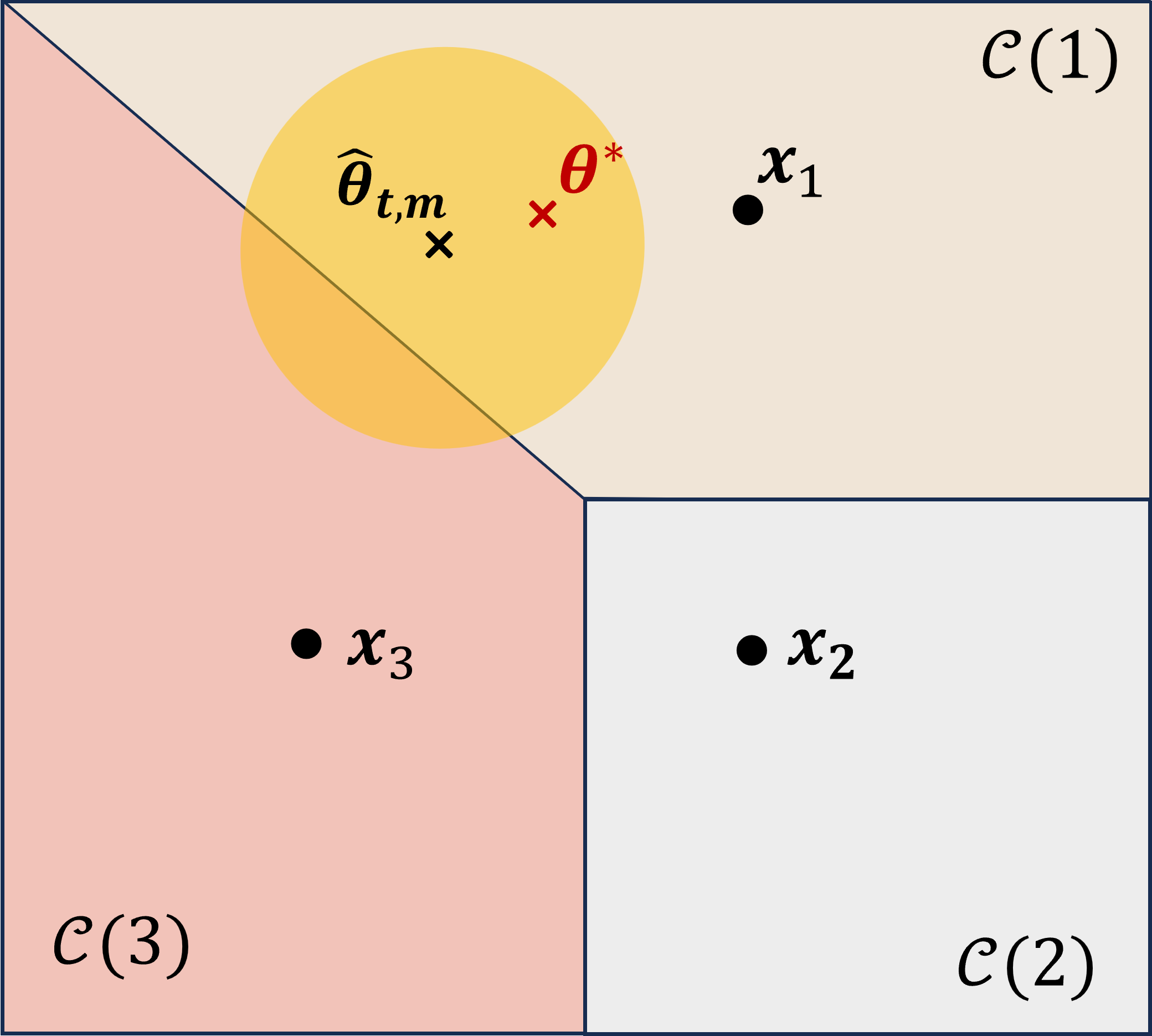}
        \caption{The confidence set at time $t$.}
        \label{fig:before_pull}
    \end{subfigure}
    \hfill
    \begin{subfigure}{0.48\columnwidth}
        \centering
        \includegraphics[width=0.8\columnwidth]{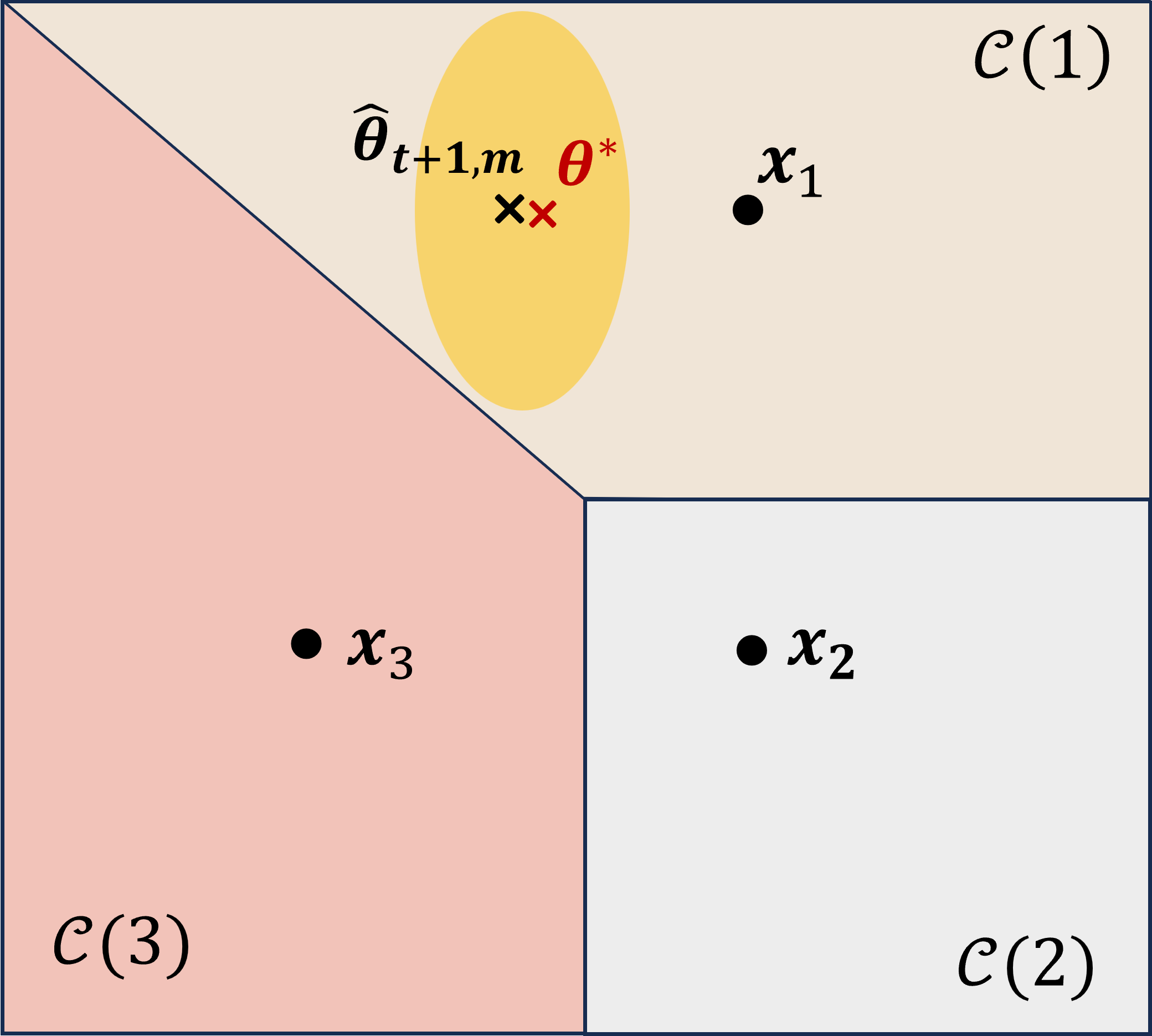}
        \caption{The confidence set at $t+1$.}
        \label{fig:after_pull}
    \end{subfigure}
    \caption{Illustration of BAI in LB.}
    \label{fig:conf_cone}
\end{figure}
\subsection{Construction of Confidence Sets}
In a single-agent linear bandit problem \cite{abbasi2011improved,soare2014best}, the agent uses locally observed information to update its estimate of the parameter vector $\boldsymbol{\theta}^{*}$. Let  $\lambda$ denote the regularization parameter and  $\boldsymbol{I}$ the identity matrix. The information includes the  $d~\times~d$ design matrix $\boldsymbol{A}_t = \lambda \boldsymbol{I} + \sum_{s=1}^{t} \X_{a_s} \X_{a_s}^\top$, and the vector $\boldsymbol{b}_t = \sum_{s=1}^{t} \X_{a_s} r_s$. Using this, the agent computes the $\ell_2$-regularized least squares estimate of $\boldsymbol{\theta}^{*}$, given by $\hat{\boldsymbol{\theta}}_{t} = \boldsymbol{A}_t^{-1} \boldsymbol{b}_t$.

However, in the collaborative algorithm, each agent $m$ has access not only to its locally observed samples but also to the additional samples collected by the coordinator from all $M$ agents. Let $t_n$ denote the start time of the $n$-th communication round. The locally updated matrix at agent $m$ at time $t$ is $ \Delta \boldsymbol{A}_{t,m} = \sum_{s =t_{n}}^t \X_{a_{s,m}}\X_{a_{s,m}}^\top$ and $\Delta \boldsymbol{b}_{t,m} =\sum_{s =t_{n}}^t  \X_{a_{s,m}}{r}_{s,m} $. 
When an agent triggers a communication event, all agents share $\Delta\boldsymbol{A}_{t,m}$ and $\Delta \boldsymbol{b}_{t,m}$ with the coordinator, which aggregates them to compute 
\begin{equation}
    \Ast{t} = \sum_{m=1}^{M}  \Delta \boldsymbol{A}_{t,m}  \,\, \text{and} \,\,  \bst{t} = \sum_{m=1}^{M}  \Delta \boldsymbol{b}_{t,m}.
    \label{eq:coor_stat}
\end{equation}
Both $\Ast{t}$ and $\bst{t}$ are sent to the agents for use in the next communication round.
Consequently, each agent uses the coordinator's aggregated data and its locally collected data to compute $ \boldsymbol{A}_{t,m}$ and $ \boldsymbol{b}_{t,m}$, defined in line~\ref{line:Atm_btm} in Algorithm~\ref{alg:DistLinGapE}. Using this information, the $\ell_2$-regularized estimate of the parameter vector is given as 
\begin{equation}
        \hat{\boldsymbol{\theta}}_{t,m} =  \boldsymbol{A}_{t,m}^{-1} \boldsymbol{b}_{t,m},
    \label{eq:theta_est}
\end{equation}
with the confidence bound given in Proposition~\ref{prop:confidence_ellipsoid} which is based on \cite[Theorem 2]{abbasi2011improved}.
\begin{proposition}\label{prop:confidence_ellipsoid}
{\textbf{Confidence Ellipsoid}}
    For linear bandits with conditionally $R$-sub-Gaussian noise for some $R \ge 0$, let $\lambda >0$ and assume that $\| \boldsymbol{\theta}^{*} \|_2 \le S$. Then, for any $\delta_m >0$ the confidence set constructed from $\boldsymbol{A}_{t,m}$ for $t \in \{1, 2, \dots \}$ is
    \begingroup
    \footnotesize
        \begin{equation}
            \label{eq:conf_ellipsoid}
            \mathscr{C}_{t,m} = \left\{ \boldsymbol{\theta} \in \Rd: \| \hat{\boldsymbol{\theta}}_{t,m} - \boldsymbol{\theta}^{*} \|_{\boldsymbol{A}_{t,m}} \le  
            R \sqrt{2 \log \frac{\det (\boldsymbol{A}_{t,m})^{\frac{1}{2}}}{\lambda^{\frac{d}{2}} \delta_m}} + \lambda^\frac{1}{2} S  \right\}.
        \end{equation}
    \endgroup
    Also, the following statement holds for $\X \in \Rd$
    \begin{equation} \label{eq:reward_diff}
            \lvert \boldsymbol{x}^{\top} (\hat{\boldsymbol{\theta}}_{t,m} - \boldsymbol{\theta}^{*}) \rvert \le \|\boldsymbol{x} \|_{\boldsymbol{A}_{t,m}^{-1}} C_{t,m},
    \end{equation}
    where 
    \begin{equation} \label{eq:conf_reward}
        C_{t,m} = R \sqrt{2 \log \frac{\det (\boldsymbol{A}_{t,m})^{\frac{1}{2}}}{\lambda^{\frac{d}{2}} \delta_m}} + \lambda^\frac{1}{2} S. 
    \end{equation}
\end{proposition}
\begin{lemma}
For any $\delta_m >0$, $\boldsymbol{\theta}^{*}$ lies in the set $\mathscr{C}_{t,m}$ with probability $1-M\delta_m$, for all $t \in \{1, 2, \dots \}$ and all $m \in [M]$.
\label{lemma:error_prob}
\end{lemma}
\begin{proof}
    This follows from Proposition~\ref{prop:confidence_ellipsoid} and a union bound over all agents.
\end{proof}
From Lemma~\ref{lemma:error_prob}, we conclude that to ensure an output with an error of at most $\delta$ in \eqref{eq:stop_cond} at each agent, we set $\delta = M \delta_m$.
\subsection{The Arm Selection Strategy}
The arm selection strategy at an agent aims to find the allocation $\x_t(m)= (\boldsymbol{x}_{a_{1,m}}, \dots, \boldsymbol{x}_{a_{t,m}})$ that identifies the best arm in the fewest number of rounds. To this end, at time $t$, each agent selects (without pulling) the arm with the maximum estimated reward, $i_{t,m}$, and the most ambiguous arm, $j_{t,m}$, given in lines \ref{line:local_best} and \ref{line:local_ambiguous} of Algorithm~\ref{alg:DistLinGapE}, respectively. Then, the agent pulls the arm that provides the largest information gain for estimating the  gap of expected rewards $( \boldsymbol{x}_{i_{t,m}} - \boldsymbol{x}_{j_{t,m}})^\top \boldsymbol{\theta}^{*}$. Geometrically, this shrinks the confidence set until it is contained within a single cone, thereby identifying the optimal arm \cite{xu2018fully}. 

In every round $t$ of Algorithm~\ref{alg:DistLinGapE}, each agent $m$ uses $\boldsymbol{A}_{t,m}$ and $\boldsymbol{b}_{t,m}$, given in line~\ref{line:Atm_btm}, to select the pulling direction specified by the \textsc{Select-Direction} procedure. The estimated gap between arms $i_{t,m}$ and $j_{t,m}$ is  
\begin{equation}
    \hat{\Delta}_{t,m}(i,j) =(\boldsymbol{x}_{i} - \boldsymbol{x}_{j})^{\top} \hat{\boldsymbol{\theta}}_{t,m},
    \label{eq:gap_est}
\end{equation}
and the confidence of the estimated gap is
\begin{equation}
    \beta_{t,m}(i,j) =  \|\boldsymbol{x}_{i} - \boldsymbol{x}_{j} \|_{\boldsymbol{A}_{t,m}^{-1}} C_{t,m},
    \label{eq:gap_bound}
\end{equation}
where $C_{t,m}$ is given in \eqref{eq:conf_reward}. 

In each time $t$, all agents check if the stopping condition is met, that is, if the upper confidence bound on the gap of the expected rewards falls below the target accuracy $\epsilon$,
\begin{equation}
    B_m(t)= \hat{\Delta}_{t,m}(j_{t,m},i_{t,m}) + \beta_{t,m}(j_{t,m},i_{t,m}) \le \epsilon.
\end{equation}
If the condition is satisfied, the algorithm stops and the arm with the highest estimated reward is identified as the best arm. Otherwise, the algorithm proceeds and each agent selects arm $a_{t,m}$ that results in the highest decrease in $ \|\boldsymbol{x}_{i_{t,m}} - \boldsymbol{x}_{j_{t,m}} \|_{\boldsymbol{A}_{t,m}^{-1}}$ and consequently in the confidence bound $ \beta_{t,m}(i_{t,m},j_{t,m})$. 

Following \cite{soare2014best,xu2018fully}, there are two  approaches for the arm selection strategy. In the greedy approach, each agent pulls the arm that minimizes the confidence bound in the current round, given by
\begin{equation}
     a_{t+1,m} = \underset{a \in [K]}{\arg \min} \|\boldsymbol{x}_{i_{t,m}} - \boldsymbol{x}_{j_{t,m}}\|_{(\boldsymbol{A}_{t,m} +\boldsymbol{x}_a \boldsymbol{x}_a^\top)^{-1}}.
     \label{eq:arm_greedy}
\end{equation}
Although we do not analyze the theoretical guarantees of this approach, we empirically show that it performs well in the numerical results \cite{xu2018fully}. 

Alternatively, the second approach involves finding the optimal ratio for arm $k$ appearing in the optimal allocation sequence $\x_t^{*}(m)$  as $t~\rightarrow~\infty$ \cite{xu2018fully}. This ratio minimizes $ \|\boldsymbol{x}_{i_{t,m}} - \boldsymbol{x}_{j_{t,m}}\|_{\boldsymbol{A}_{t,m}^{-1}}$. 
Let $ w_k^{*}(i_{t,m}, j_{t,m}) $ be the $k$-th element in the vector $\boldsymbol{\mathrm{w}}^{*} (i_{t,m}, j_{t,m})$ defined as
\begin{align}
    &\boldsymbol{\mathrm{w}}^{*} (i_{t,m}, j_{t,m}) = \underset{  \boldsymbol{\mathrm{w}} \in \Rd}{\arg \min} |\boldsymbol{\mathrm{w}}|, \\
    &\text{s.t.} \quad \boldsymbol{x}_{i_{t,m}} - \boldsymbol{x}_{j_{t,m}} = \sum_{k=1}^{K} w_k \boldsymbol{x}_k.
\end{align}

The optimal ratio for selecting arm $k$ is given by:
\begin{equation}
    p^{*}_k (i_{t,m}, j_{t,m}) = \frac{\lvert w_k^{*}(i_{t,m}, j_{t,m}) \rvert}{\sum_{k=1}^{K} \lvert w_k^{*}(i_{t,m}, j_{t,m}) \rvert }.
\end{equation}
Let $T_{a,m}(t)$ denote the number of times arm $a$ is selected up to time $t$. This includes the number of times arm $a$ was selected by other agents in the past communication rounds and the number of times the arm is selected locally until time $t$. The arm selection strategy aims to keep the actual arm selection ratio as close as possible to the target ratio. Thus, the pulled arm satisfies
\begin{equation}
     a_{t+1,m} = \underset{a \in [K]: p^{*}_a (i_{t,m}, j_{t,m}) >0}{\arg \min} \quad \frac{T_{m,a}(t)}{p^{*}_a (i_{t,m}, j_{t,m})}.
\label{eq:arm_selection}
\end{equation}

\begin{theorem}
\label{th:opt_arm}
    The arm $\hat{a}^{*}_m$ returned by the DistLinGapE using the arm selection strategy in \eqref{eq:arm_greedy} or \eqref{eq:arm_selection} satisfies \eqref{eq:stop_cond}.
\end{theorem}
\begin{proof}
    The proof is presented in Appendix~\ref{app:confidence} in the supplemantal material.
\end{proof}
\begin{algorithm}[ht]
\caption{The DistLinGapE Algorithm}
\label{alg:DistLinGapE}
\begin{algorithmic}[1]
    \STATE \textbf{Input:}  $\epsilon$, $\delta$, $S$, $R$, $\rho$
 and $\lambda$
    \STATE \textbf{Output} Arm $\hat{a}^*$ that satisfies stopping condition.
    \STATE \textbf{Initialize agents}:
     $\forall m~\in~[M]$, $\boldsymbol{A}_{0,m}~=~\boldsymbol{0}$, $\boldsymbol{b}_{0,m}~=~\boldsymbol{0}$, $\Delta \boldsymbol{A}_{0,m}=\boldsymbol{0}$, $\Delta \boldsymbol{b}_{0,m}=\boldsymbol{0}, \Delta t_{0,m}=0$. For strategy \eqref{eq:arm_selection} initialize $\Delta T_{m,k}=0,  T_{m,k}=0 $ and $T_k=0 \, \forall k \in [K]$
    \STATE  \textbf{Initialize coordinator}: Set $\Ast{0}=\boldsymbol{0}$, $\boldsymbol{b}_{\text{coor},0}=\boldsymbol{0}$
    \FOR{$t=1, \dots, \infty$}
        \FOR{Agent $m=1, \dots, M$}
             \STATE \label{line:begin_local_update}
             $\boldsymbol{A}_{t,m} = \lambda \boldsymbol{I} +\Ast{t} + \Delta \boldsymbol{A}_{t,m}$, $\boldsymbol{b}_{t,m} = \boldsymbol{b}_{\text{coor},t} + \Delta \boldsymbol{b}_{t,m}$, and $ \Delta t_{t,m}  +=1$. \label{line:Atm_btm}
             \STATE Select  pulling direction: $(i_{t,m}, j_{t,m}, B_m(t)) \leftarrow \textsc{Select-Direction} (\boldsymbol{A}_{t,m}, \boldsymbol{b}_{t,m})$
            \IF{$B_m(t) \leq \epsilon$} \label{line:stopping_condition}
                \STATE $i_{t,m}$ is the best arm $\hat{a}^*$; Terminate algorithm
            \ENDIF
            \STATE \label{line:arm_pull} Pull arm $a_{t,m}$ according to \eqref{eq:arm_greedy} or \eqref{eq:arm_selection} 
            \STATE Observe $r_{t,m}$

            \STATE Update $\Delta \boldsymbol{A}_{t,m} += \boldsymbol{x}_{a_{t,m}}             \boldsymbol{x}_{a_{t,m}}^\top$, \\ $\Delta \boldsymbol{b}_{t,m} += \boldsymbol{x}_{a_{t,m}} r_{t,m}$. For strategy \eqref{eq:arm_selection} $\Delta T_{m,a_t}(t) +=1 $ and $T_{m,a_t}(t)=T_{a_t}(t) + \Delta T_{m,a_t}(t)$  \\
         // Communication condition
        \IF{$\Delta t_{t,m}\log \frac{\det (\boldsymbol{A}_{t,m})} { \det(\lambda \boldsymbol{I} + \Ast{t} )} > D$} \label{line:comm_condition}
            \STATE Each agent $m \in [M]$ sends $\Delta \boldsymbol{A}_{t,m}, \Delta \boldsymbol{b}_{t,m}$, and $\Delta T_{m,k}(t), \forall k \in [K] $ (for strategy  \eqref{eq:arm_selection}) to coordinator.
            \STATE Agents reset $\Delta \boldsymbol{A}_{t,m}=0, \Delta \boldsymbol{b}_{t,m}=0$, $\Delta t_{t,m}=0$, and $\Delta T_{m,k}(t)=0$.
            \STATE Coordinator collects $\Delta \boldsymbol{A}_{t,m}$, $\Delta \boldsymbol{b}_{t,m}$ and $\Delta T_{m,k}(t) \, \forall k \in [K]$ 
            \STATE Coordinator computes $\Ast{t} + = \sum_{m=1}^{M} \Delta \boldsymbol{A}_{t,m}$, \\ 
            $\boldsymbol{b}_{\text{coor},t} += \sum_{m=1}^{M} \Delta \boldsymbol{b}_{t,m}$ and  $T_k(t)=\sum_{m=1}^{M} \Delta T_{m,k}(t), \forall k \in [K]$
            \STATE Coordinator broadcasts $\Ast{t}, \boldsymbol{b}_{\text{coor},t}$ and $T_k(t), \forall k \in [K]$ to all agents
            \ENDIF        
        \ENDFOR
    \ENDFOR \\ \vspace{-0.3cm}
    \rule{\linewidth}{0.1pt}  \vspace{-0.5cm}
    \STATE \textbf{Procedure} {Select-Direction}($\boldsymbol{A}_{t,m}, \boldsymbol{b}_{t,m}$)
        \STATE \hspace{1em}$\hat{\boldsymbol{\theta}}_{t,m} \leftarrow \boldsymbol{A}^{-1}_{t,m}\boldsymbol{b}_{t,m}$ \label{line:est_theta_agent}
        \STATE \hspace{1em}$i_{t,m} \leftarrow \arg\max_{i\in[K]} (\boldsymbol{x}_i^\top\hat{\boldsymbol{\theta}}_{t,m})$ \label{line:local_best}
        \STATE \hspace{1em}$j_{t,m} \leftarrow \arg\max_{j\in[K]} (\hat{\Delta}_{t,m}(j,i_{t,m}) + \beta_{t,m}(j,i_{t,m}))$ \label{line:local_ambiguous}
        \STATE  \hspace{1em}$B_m(t) \leftarrow \max_{j\in[K]}  (\hat{\Delta}_{t,m}(j,i_{t,m}) + \beta_{t,m}(j,i_{t,m}))$
        \STATE  \hspace{1em}\textbf{return} $i_{t,m}, j_{t,m},B_m(t)$
    \STATE \textbf{EndProcedure}
    \end{algorithmic}
\end{algorithm}
%
%
\subsection{DistLinGapE Sample Complexity}
Using the arm selection strategy in \eqref{eq:arm_selection}, we derive an upper bound on the sample complexity of the proposed DistLinGapE algorithm and show that agent collaboration significantly reduces it.

The sample complexity is expressed in terms of the problem complexity, defined as follows \cite{xu2018fully}
\begin{equation}
    H_\epsilon \coloneqq \sum_{k=1}^K \underset{i,j \in [K]}{\max} \frac{p^{*}_k (i, j) \alpha(i,j)}{\max \left(\epsilon, \frac{\epsilon+\Delta_i}{3} , \frac{\epsilon+\Delta_j}{3}\right)^2},
    \label{eq:H}
\end{equation}
where  $\alpha(i,j)= |\boldsymbol{\mathrm{w}}^{*}(i,j)| $ and  $\Delta_i = (\X_{a^{*}}-\X_i)^\top \boldsymbol{\theta}^{*}$ for $a^{*} \neq i$. The value of $H_\epsilon$ captures the order of the minimum number of samples required to distinguish between arms $i$ and $j$ with accuracy $\epsilon$.

The improvement in sample complexity of the collaborative algorithm compared to the centralized (single-agent) algorithm is measured by the algorithm speedup. Let $T_{\mathcal{O}}$ be the per-agent sample complexity of the best centralized algorithm and $T_{\mathcal{A}}$ be the per-agent sample complexity of the proposed collaborative algorithm $\mathcal{A}$. The speedup is defined as \cite{tao2019collaborative}
\begin{equation}
    S_{\mathcal{A}} = T_{\mathcal{O}}/T_{\mathcal{A}}.
\end{equation}
The speedup can take any value between $1 \le S_{\mathcal{A}} \le M$. It is $1$ if the agents learn in isolation and $M$ is the optimal speedup.

To understand the speedup $S_{\mathcal{A}}$, we note that in the multi-agent setting, the samples required to achieve an accuracy of $\epsilon$ in \eqref{eq:H} are collected across $M$ agents. Ideally, this would lead to a per-agent sample complexity that is $1/M$ of the single-agent case, i.e., $T_{\mathcal{A}} = T_{\mathcal{O}}/M$. However, due to the randomness in the observed rewards, which influences the triggering of communication events, the actual per-agent sample complexity can be higher, with $T_{\mathcal{A}} = T_{\mathcal{O}}/S_{\mathcal{A}} $. 

Theorem \ref{th:sample_complexity} gives the upper bound on the per-agent sample complexity.
\begin{theorem}
\label{th:sample_complexity}
    The DistLinGapE algorithm can identify the $(\epsilon, M\delta_m)$-best arm using the arm selection strategy in \eqref{eq:arm_greedy} or \eqref{eq:arm_selection} with probability $1-M \delta_m$ and the following sample complexity per agent:
    \begin{itemize}
        \item  For $\lambda \leq \frac{2R^2}{S^2}\log\frac{K^2}{\delta_m}$: 
        \begin{equation}
            \tau_m \leq \mu + 4 \Hm R^2 \left(
            2 \log\frac{K^2}{\delta_m}  +
            d\log\left(1+\frac{Y^2 L^2}{\lambda d}\right),
    \right) 
    \end{equation}
    where $\mu = \frac{K}{M}+1$, $Y = 2\sqrt{16  H_\epsilon^2 R^4 d L^2/(M\lambda) +N^2}$, and $N=\frac{8 H_\epsilon R^2}{M} \log \frac{K^2}{\delta_m}$.
    \item For $\lambda > 4  H_\epsilon R^2L^2$:
    \begin{equation}
        \tau_m \le 2 \left( \frac{4 H_\epsilon R^2}{M} \log\frac{ K^2}{\delta_m} +  \frac{2 H_\epsilon \lambda S^2}{M} + \mu \right) 
    \end{equation}

  \end{itemize}
\end{theorem}
\begin{proof}
    The proof of Theorem~\ref{th:sample_complexity} is given in Appendix~\ref{app:sample_complexity} in the supplemental material.
\end{proof}

\subsection{Communication Rounds}
\label{sec:comm_rounds}
We measure the communication cost as the number of communication rounds in the learning process. To minimize this cost, an agent $m$ initiates a communication round only when there is a significant change in the matrix $\boldsymbol{A}_{t,m}$. 
Since the volume of the confidence ellipsoid given by \eqref{eq:conf_ellipsoid} depends on $\det (\boldsymbol{A}_{t,m})$, where $\det(\cdot)$ is the matrix determinant, this change is evaluated by the ratio $\frac{\det (\boldsymbol{A}_{t,m})} { \det(\lambda \boldsymbol{I} + \Ast{t} )}$. Here, $\det(\lambda \boldsymbol{I} + \Ast{t} )$ reflects the volume at the beginning of the current communication round and $\det (\boldsymbol{A}_{t,m})$ is proportional to the volume of the confidence ellipsoid at round $t$ \cite{wang2019distributed}.  

Let $\tau$ be the total sample complexity of the DistLinGapE algorithm across all agents. The corresponding upper bound on the total number of communication rounds is given by Theorem~\ref{th:comm_cost}.
\begin{theorem}
\label{th:comm_cost} 
    The total communication cost of the DistLinGapE algorithm is upper bounded by $O\left( M \sqrt{ \frac{ \tau d \log_2 \tau }{D} }\right)$.
\end{theorem}
\begin{proof}
    The proof of Theorem \ref{th:comm_cost} is given in Appendix \ref{app:comm_cost} in the supplemental material.
\end{proof}
\subsection{Communication Threshold Selection}

The communication threshold $D$ in Algorithm~\ref{alg:DistLinGapE} controls the
tradeoff between communication overhead and convergence speed. Communication
with the coordinator is triggered only when the condition in
Line~\ref{line:comm_condition} of Algorithm~\ref{alg:DistLinGapE} is satisfied.

Let $T$ denote the (random) stopping time of DistLinGapE measured in global
rounds, where a global round corresponds to an iteration at which communication
with the coordinator may be triggered. Between two consecutive global rounds,
each agent may perform multiple local updates and collect multiple samples.
Let $\tau_m$ denote the total number of local samples collected by agent $m$ up
to stopping time $T$, and let
$\tau = \sum_{m=1}^{M} \tau_m$ denote the total sampling complexity across all agents.

From Theorem~\ref{th:sample_complexity}, each agent has a deterministic upper
bound $\tau_m \le \bar{\tau}_m$, which implies $
\tau \le \sum_{m=1}^{M} \bar{\tau}_m.$
Note that, in general, $\tau_m \ge T$, since agents may collect more than one
local sample between communication events.

Since at most one synchronization event can be triggered per global round, and
each synchronization consists of $2M$ transmissions, the
total number of communications up to stopping time $T$ satisfies
\begin{equation}
    C(T) \le 2MT \qquad \text{almost surely}.
\end{equation}

We introduce a communication rate parameter $\rho \in (0,1]$ to control the
communication overhead. The communication budget is defined as a fraction
$\rho$ of the number of global rounds and is given by
\begin{equation}
    B_{\mathrm{c}} = \big\lceil \rho\,T \big\rceil,
\end{equation}
which corresponds, on average, to one synchronization every $1/\rho$ global
rounds.

Using the communication bound in Theorem~\ref{th:comm_cost},
\[
C(T)
=
O\!\left(
M\sqrt{\frac{\tau\, d\log_2(\tau)}{D}}
\right),
\]
we choose the threshold $D$ by inverting the bound so that the total number of
communications satisfies $C(T) \le B_{\mathrm{c}}$, which yields
\begin{equation}
    D
    =
    \frac{M^{2}\,\tau\, d\log_2(\tau)}{B_{\mathrm{c}}^{\,2}}.
\end{equation}
With this choice of $D$, the total communication cost satisfies
$C(T) = O(B_{\mathrm{c}})$ in the worst case.

\subsection{Computational Complexity of DistLinGapE}

The DistLinGapE runs in polynomial time, which makes it suitable for problems with a moderate number of arms and context dimension $d$. Consequently, it is well suited our service placement problem with a moderate number of SBSs, services, and context dimensions. 

The dominant computational cost of the algorithm comes from evaluating the confidence terms at the agents for arm selection, as discussed below. 

\begin{proposition}
\label{prop:distlingape_comp}
Consider DistLinGapE with $K$ arms and $d$-dimensional context vectors. By using the inverse matrix lemma to obtain the local inverse design matrix $\boldsymbol{A}_{t,m}^{-1}$,
the computational complexity per round at each agent $m$ is $O(Kd^2)$. Consequently, the total computational complexity is $O(M \tau_m K d^2) +O\left(N_{\text{comm}}(Md^2 +MK)\right)$,
where $N_{\text{comm}}$ is the number of communication rounds. 
\end{proposition}

\begin{proof}
At iteration $t$ and agent $m$, DistLinGapE performs the following steps. First, it updates $\boldsymbol{A}_{t,m}^{-1}$ and $\hat{\boldsymbol{\theta}}_{t,m}$ using the inverse matrix lemma, which costs $O(d^2)$. Then, it identifies the current best arm which costs $O(Kd)$. Then, it identifies the most ambiguous arm by evaluating the confidence terms  $
\left\|\boldsymbol{x}_j - \boldsymbol{x}_{i_{t,m}}\right\|_{\boldsymbol{A}_{t,m}^{-1}}^2,
$ with a cost $O(d^2)$. Evaluating this quadratic form for all $K$ candidate arms results in a total cost of $O(Kd^2)$, which dominates the other values. Hence, the  computational complexity at each agent and each round is $O(Kd^2)$, and the cost over all iterations is $O(\tau_m Kd^2)$. 

In each global round, the coordinator aggregates the values of $\Delta \boldsymbol{A}_{t,m}$ at a cost of $O(M d^2)$, and $\Delta \boldsymbol{b}_{t,m}$ at a cost of $O(M d)$, and the $\Delta T_{m,k}, \forall k\in [K]$ at a cost $O(MK)$. So the aggregation cost over $N_{\text{comm}}$ rounds is $O\left(N_{\text{comm}}(Md^2 +MK)\right)$
\end{proof}

\section{Extensions to Heterogeneous Agents}
\label{sec:extension}

The DistLinGapE algorithm proposed in Section~\ref{sec:DistLinGapE} and the corresponding theoretical analysis consider a homogeneous multi-agent setting, where all agents have access to the same set of arms (services) and collaboratively identify a common best arm. This model addresses an important cooperative service placement problem for SBSs operating in homogeneous environments, such as hospitals, industrial facilities, or campuses. It could be also used in more diverse environments with a subset of homogeneous SBSs, as collaboration among them accelerates learning. However, practical communication systems can be heterogeneous across agents, while sharing the same parameter vector $\boldsymbol{\theta}^*$.

Heterogeneity across SBSs can result from differences in user density or local demand patterns. This results in variations in context vectors, which affect the observed rewards of the arms. In this section, we discuss a heterogeneous case that can be address, to some extent, using the DistLinGapE algorithm. It  considers agents with different sets of context vectors but a common reward ordering and share the same globally optimal arm. This setting models SBSs operating under different user densities but similar service preferences. Then, we discuss a possible extension for the DistLinGapE that considers agents with different contexts and potentially different reward orderings, where the objective is to identify the best arm locally at each agent, corresponding to SBSs with heterogeneous user preferences and densities.

\subsection{Unequal User Density: Agents with a Common Arm Ordering}

We consider the case in which SBSs may have different user densities but similar service preferences. Since the context vector used for service placement decision is constructed from long-term average service demand, the context vectors across SBSs differ only by a positive scaling factor. As a result, although the reward values vary across SBSs, the underlying arm ordering structure determined by the context vectors remains unchanged across agents.

In this setting, the DistLinGapE algorithm remains valid, as all agents have access to the same set of arms, corresponding to the same set of services, and the heterogeneity across SBSs is reflected only in the context vectors. To explain this, let the context vector of service $k$ at SBS $m$ be given by
\[
\boldsymbol{x}_{k,m} = c_m \boldsymbol{x}_k, \quad c_m > 0,
\]
where $c_m$ reflects the user density at SBS $m$, and $\boldsymbol{x}_k$ represents the contextual information of service $k$. For any two services $k$ and $\ell$, the ordering of their expected rewards at SBS $m$ is determined by
\[
(\boldsymbol{x}_{k,m} - \boldsymbol{x}_{\ell,m})^\top \boldsymbol{\theta}^* \ge 0,
\]
which is equivalent to
\[
(\boldsymbol{x}_{k} - \boldsymbol{x}_{\ell})^\top \boldsymbol{\theta}^* \ge 0.
\]
As a result, the resulting partition of the parameter space into cones remains identical across all SBSs. Consequently, all $M$ agents solve the same BAI problem in terms of space partitioning and arm comparisons. The BAI process therefore follows Algorithm~\ref{alg:DistLinGapE}, with the scaling factors $c_m$ affecting only the magnitude of the observed rewards and the rate at which confidence sets shrink, so the sample complexity and communication bounds derived for homogeneous contexts continue to hold, up to constant scaling factors.

\subsection{Heterogeneous Demand: Local Best-Arm Identification}

In a more general heterogeneous setting, SBSs may serve user populations with different service preferences, resulting in demand vectors that differ beyond simple scaling. Consequently, the reward ordering across services may differ across SBSs, even though rewards remain linear in the context with a shared parameter vector $\boldsymbol{\theta}^*$.

In this case, the context vectors associated with the arms (services) differ across agents and the separating hyperplanes defined by
\[
(\boldsymbol{x}_{k,m} - \boldsymbol{x}_{\ell,m})^\top \boldsymbol{\theta} = 0
\]
may no longer coincide across SBSs. This leads to different  partitions of the parameter space into cones across agents and different local best arms. In this case, the learning objective shifts from identifying a single global best service to identifying the local best service at each SBS.

Despite this heterogeneity, collaboration in learning the shared $\boldsymbol{\theta}^*$ remains beneficial. By exchanging information, SBSs can improve the estimation accuracy of $\boldsymbol{\theta}^*$ and accelerate confidence shrinkage. However, the stopping and recommendation decisions must be defined locally at each SBS. 

\section{Robustness to Communication Failures}
\label{sec:robustness}
The proposed DistLinGapE algorithm assumes reliable communication; however, practical networks are prone to communication failures and  delays. In these cases, the algorithm degrades gracefully. When updates between SBSs and the coordinator are delayed or unavailable, each SBS continues updating its local statistics $\boldsymbol{A}_{t,m}$ and $\boldsymbol{b}_{t,m}$ based on local observations. This reduction in information leads to slower convergence and increased sample complexity, while preserving the correctness of the stopping condition and the identification of the best arm. In the extreme case of full communication failure, the algorithm reduces to independent LinGapE learning at each SBS.

\begin{lemma}
Consider the DistLinGapE algorithm under possible communication failures.
For any SBS $m$, round $t$, and any pair of arms $(i,j)$, the estimated gap under full communication and its confidence are given in \eqref{eq:gap_est} and \eqref{eq:gap_bound}, respectively. Let $\tilde{\beta}_{t,m}(i,j)$ be the confidence  constructed using the potentially partial aggregated statistics available due to communication failures.
Then, with probability at least $1-\delta$, it holds that
\begin{equation}
    \tilde{\beta}_{t,m}(i,j)
\;\ge\;
\beta_{t,m}(i,j),
\quad \forall t,\ \forall m \in [M],\ \forall i,j \in [K].
\end{equation}
Consequently, communication failures do not reduce the confidence bounds on the estimated gaps and therefore cannot cause early stopping or false identification of the best arm, they can only increase the stopping time.
\end{lemma}

\begin{proof}[Proof sketch]
At round $t$, let $\boldsymbol{A}_{\text{coor},t}$ given in \eqref{eq:coor_stat}
be the aggregated updates at the coordinator under full communication.
Under communication failures, the coordinator may only receive updates from a subset of SBSs, yielding
\begin{equation}
\tilde{\boldsymbol{A}}_{\text{coor},t}
=
\sum_{j \in \mathcal{S}_{t}} \Delta \boldsymbol{A}_{t,j},
\end{equation}
where $\mathcal{S}_{t} \subseteq \{1,\ldots,M\}$ is the set of SBSs whose updates are successfully received by the coordinator at round $t$.

Since each $\Delta \boldsymbol{A}_{t,j}$ is positive semidefinite, then 
$
\tilde{\boldsymbol{A}}_{\text{coor},t}
\preceq
\boldsymbol{A}_{\text{coor},t}$. The local design matrix constructed at SBS $m$ under communication failures is $\tilde{\boldsymbol{A}}_{t,m}
=\lambda \boldsymbol{I}
+\tilde{\boldsymbol{A}}_{\text{coor},t}
+\Delta \boldsymbol{A}_{t,m}.$
Therefore,
$\tilde{\boldsymbol{A}}_{t,m}
\preceq
\boldsymbol{A}_{t,m}$ implies 
$\tilde{\boldsymbol{A}}_{t,m}^{-1}
\succeq
\boldsymbol{A}_{t,m}^{-1}.$

For any pair of arms $(i,j)$, it follows that
$\|\boldsymbol{x}_i - \boldsymbol{x}_j\|_{\tilde{\boldsymbol{A}}_{t,m}^{-1}}
\ge
\|\boldsymbol{x}_i - \boldsymbol{x}_j\|_{\boldsymbol{A}_{t,m}^{-1}},$ 
which  implies
\begin{equation}
    \tilde{\beta}_{t,m}(i,j)
\ge
\beta_{t,m}(i,j).
\end{equation}

Larger confidence bounds due to communication failures cannot satisfy the stopping condition early. Thus, communication failures may delay stopping but cannot lead to false identification of the best arm.
\end{proof}
%

\section{Numerical Results}
\label{sec:numerical_results}
In this section, we present and analyze the performance of the proposed algorithm, first using synthetic data and then through simulations of a small cell network. We compare the performance of the proposed DistLinGapE algorithm with the following baselines:
\begin{enumerate}
    \item \textbf{$\mathcal{XY}$-Oracle Strategy} \cite{soare2014best}: This assumes that $\boldsymbol{\theta}^{*}$ is known to the learner and is used as a lower bound on the sample complexity. 

    \item \textbf{$\mathcal{XY}$-Adaptive Strategy} \cite{soare2014best}: It is a semi-adaptive algorithm that divides the total number of rounds into phases, and eliminates directions that do not contribute to reducing the uncertainty in estimating the parameter vector $\boldsymbol{\theta}^{*}$ at the end of each phase. Within each phase, a fixed allocation strategy is employed.
    \item \textbf{LinGapE} \cite{xu2018fully}: It is a fully adaptive approach that adjusts the arm selection based on the rewards observed in the previous round. This work is limited to the single-agent setting and serves as the main algorithm on which our work is based, extending it to the multi-agent case.
\end{enumerate}
All results are averaged over 30 algorithm runs. The confidence $\delta_m$ is set to $0.05$.
\subsection{Simulation on Synthetic Data}
Here we use the benchmark example from \cite{soare2014best}, commonly used in BAI in LB. This example consists of $d+1$ arms, where $d$ is the dimension of the context vector. The context vectors of the first $d$ arms are the canonical basis vectors, given by $\X_1 = \boldsymbol{e}_1, \dots, \X_d = \boldsymbol{e}_d$. The context vector of the last arm is $\boldsymbol{x}_{d+1}=[\cos(\phi), \sin(\phi), 0, \dots, 0]^\top$ with $\phi =0.01$.  The parameter vector is $\boldsymbol{\theta}^{*} = [2, 0, \dots, 0]^\top$ and the noise  follows $\eta_{t,m} \sim \mathcal{N}(0,1)$.  Since $\arg \max_{k \in [K]} \boldsymbol{x}_k^\top \boldsymbol{\theta}^{*}=1$, arm 1 is the optimal arm $a^{*}$ in this example. 
\begin{figure}[ht]
    \centering
    \includegraphics[width=0.9\linewidth]{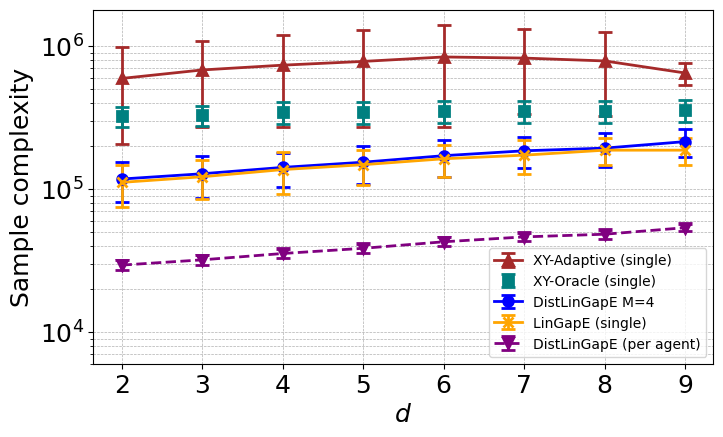}
    \caption{The sample complexity performance for different values of $d$.}
    \label{fig:eg_d_vs_complexity}
\end{figure}

Fig.~\ref{fig:eg_d_vs_complexity} shows how the sample complexity changes with the change in $d$. The figure compares the proposed DistLinGapE algorithm with $M=4$ agents with the single-agent (centralized) LinGapE, as well as the $\mathcal{XY}$-Oracle and the $\mathcal{XY}$-Adaptive algorithms. 
DistLinGapE, with $M=4$ agents, achieves a total sample complexity close to that of LinGapE, effectively reaching the optimal speedup of $M$. On average, each agent contributes an equal number of arm pulls, resulting in a per-agent sample complexity of $1/M$ of the total sample count, as shown by the dashed line. In contrast, the $\mathcal{XY}$-Adaptive algorithm, being only semi-adaptive, does not use past observations at every time step. Consequently, it performs worse than fully adaptive algorithms, requiring a larger number of samples to achieve the same confidence level specified as $\delta = M \delta_m = 0.2$. Despite knowing $\boldsymbol{\theta}^{*}$,  $\mathcal{XY}$-Oracle has higher sample complexity than DistLinGapE because DistLinGapE has tighter confidence bounds, which enable more accurate reward estimates \cite{xu2018fully}.
\begin{table}[ht]
    \caption{Number of Samples per Arm for $d=5$.}
    \centering
    \rowcolors{2}{gray!20}{white}
    \begin{tabular}{|c|c|c|c|c|} \hline
     \textbf{Arm }   & \textbf{DistLinGapE} & \textbf{LinGapE} & \textbf{$\mathcal{XY}$-Adaptive} & \textbf{$\mathcal{XY}$-Oracle}  \\ \hline
     Arm 1  &  794.36    & 727.77     & 3898.20   &   1623.73   \\
      Arm 2  &  153060.47 & 147117.53  & 776540.00 & 340027.27 \\
     Arm 3  &  34.07     & 27.10       & 30.27  &   1618.00   \\
     Arm 4  &  40.43     &  25.47      & 18.30  &  1620.10   \\
     Arm 5  &  33.10     &  28.90      & 20.73  &  1604.80   \\
     Arm 6  &  8.47      &  4.97       & 20.53   &  1616.67 \\ \hline
    \end{tabular}
    \label{tab:pull_per_arm}
\end{table}

Table~\ref{tab:pull_per_arm} shows the number of samples per arm for $d=5$. The DistLinGapE column represents the number of samples across all four agents. Arm 2 is pulled significantly more than any other arm, including the optimal arm 1.
This occurs because, for small values of $\phi$, $0 < \phi \ll 1$, arm $d+1$ is the strongest competitor to arm 1, with $\Delta_{\min} = (\boldsymbol{x}_1 - \boldsymbol{x}_{d+1})^{\top} \boldsymbol{\theta}^{*}$. To accurately identify arm 1 as the best, the uncertainty in the direction $\boldsymbol{y} = (\boldsymbol{x}_1 - \boldsymbol{x}_{d+1})$ must be minimized. Notably, arm 2 is nearly aligned with this direction, making it the most effective choice for reducing the uncertainty.
The number of samples per agent in DistLinGapE is $1/M$ of the table value. In the service placement problem, this suggests that the server may repeatedly deploy a suboptimal service at the SBSs to confidently identify the best one.

\subsection{Simulation on the Small Cell Network}
We consider a network with a service provider with $K=10$ services, which can be deployed either at the cloud or the SBSs. Once the optimal service deployment is identified, it remains fixed for an extended period compared to task scheduling. The context vector for each arm $\boldsymbol{x}_k$ is an $8$-dimensional vector in this example, obtained from historical information about the the average delay improvement of a service and the long-term average demand over a four minute period.
To model the demand vectors in our simulations, the base demand of each service follows a Zipf distribution and is perturbed by multiplicative log-normal noise to represent demand fluctuations across different time periods. Fig.~\ref{fig:demand_vectors} shows the normalized average service demand for the 10 services across 8 time periods.
%
\begin{figure}[ht]
    \centering
    \includegraphics[width=0.9\linewidth]{ 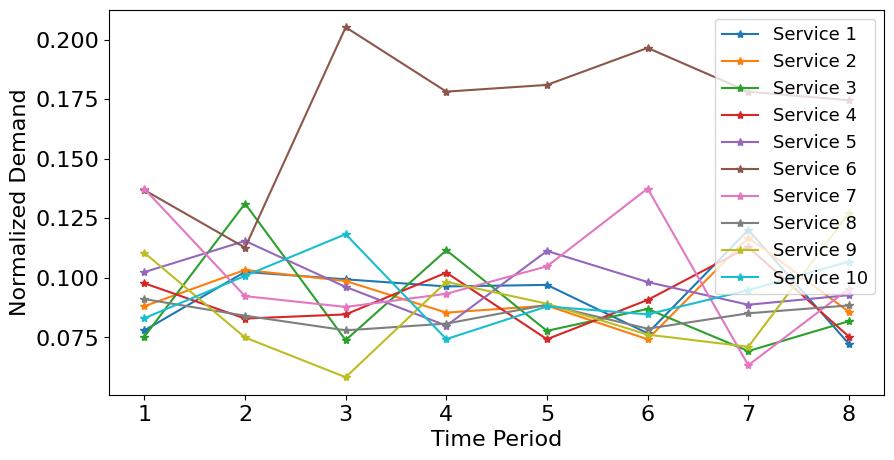}
    \caption{The normalized demand for the $K=10$ services over 8 periods of time.}
    \label{fig:demand_vectors}
\end{figure}

To study the performance of the network under different numbers of collaborating SBSs, we consider a network consisting of six small cells, each with a radius of $200$~m and an SBS located at its center. The cells are arranged in a circular layout such that their edges touch, forming an outer circle with a radius of $600$~m. An MBS is positioned at the center of this outer circle. The value of $M$ determines the number of SBSs collaborating in the learning process. Each small cell contains 75 users distributed uniformly at random within its coverage area. The main simulation parameters are summarized in Table~\ref{tab:simulation_parameters}, where parameters given as $[a,b]$ are sampled from the uniform distribution $\mathcal{U}[a,b]$.

The uplink channel gain $h_{t,m}$ in \eqref{eq:datarate} is defined as the average gain across users. In the simulation, the channel gain at time $t$ for user $u$ in SBS $m$ is given by
$h_{t,m,u} = \alpha_{t,m,u}^2 \cdot g_m \cdot \left({l_{\text{ref}}}/{l_{t,m,u}}\right)^\nu$,
where $\alpha_{t,m,u}^2$ represents the small-scale Rayleigh fading power gain, modeled as an exponential random variable with unit mean, and $g_m \cdot \left({l_{\text{ref}}}/{l_{t,m,u}}\right)^\nu$ models the path loss. Here, $l_{\text{ref}}$ is the reference distance $1$~m, $l_{t,m,u}$ is the distance between user $u$ and SBS $m$, $g_m=-30$ dB is the path-loss coefficient, and $\nu = 2.5$ is the path-loss exponent.
Similarly, the uplink channel between user $u$ and the MBS is modeled as
$\alpha_{t,0,u}^2 \cdot g_0 \cdot \left({l_{\text{ref}}}/{l_{t,0,u}}\right)^\nu$,
where $g_0 = -40$~dB, $\alpha_{t,0,u}^2$ models the small-scale Rayleigh fading power and follows an exponential distribution with mean $1$, and $l_{t,0,u}$ is the distance between user $u$ and the MBS.

\begin{table}[ht]
    \caption{Simulation parameters}
    \centering
    \begin{tabular}{|c|c|c|} \hline
        Parameter   & Description         &  Value \\ \hline \hline
        W           & Channel bandwidth   & 10 MHz  \\
        P           & Transmission power   & 10 dBm \\
        $\sigma_{\text{N}}^2$ & Noise power & $-104$~dBm \\
        $I$         &  Interference power    & $-90$~dBm \\
        $\rho^{\text{b}}$ & Backhaul data rate & $[1,4]$~GHz \\
        $RRT_{t,m}$      & Round trip time to SBS     & $[2,7]$~ms \\  
        $RRT_{t,0}$      & Round trip time to cloud     & $[20,40]$~ms \\  
        $s_k$      & task size     & $[0.5,1]$~MB \\  
        $f_{t,k,0}$  & Cloud CPU frequency   & $[4.6,5.6]$~GHz \\  
        $f_{t,k,m}$  & SBS CPU frequency   & $[2.3, 3.2]$~GHz \\  
        $c$  & Processing cycles/bit   & $100$ \\ \hline 
    \end{tabular}
    \label{tab:simulation_parameters}
\end{table}
\begin{figure}
    \centering
    \includegraphics[width=0.8\linewidth]{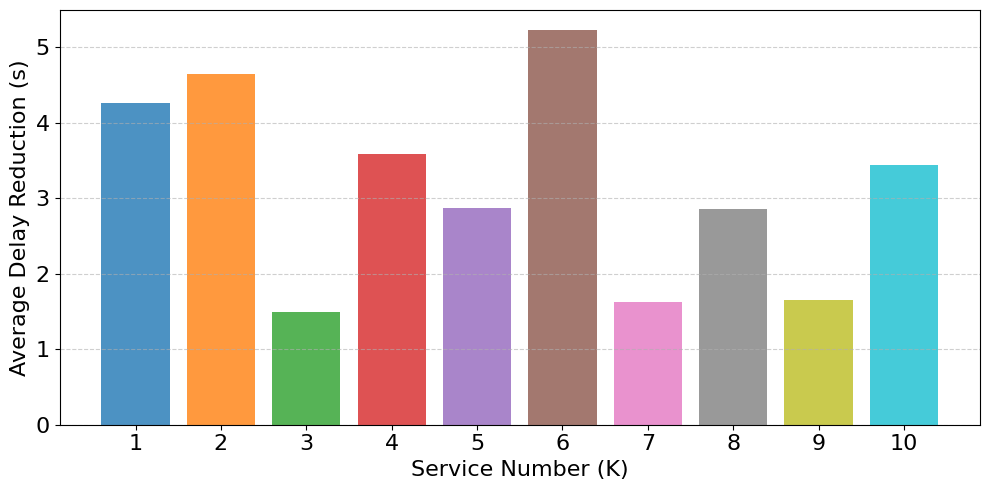}
    \caption{The average delay reduction for each service $k$ for a network with $K=10$ and $M=4$.}
    \label{fig:delays}
\end{figure}

Fig.~\ref{fig:delays} shows the average delay reduction achieved by placing service $k$ at the SBSs in a four SBS network, given that there are 75 users in each coverage area $m$. In this example, service $6$ is the best. The DistLinGapE algorithm correctly identifies service $6$ as the best with zero error, despite the unknown service demand and the randomness of the environment.

Table~\ref{tab:speedup_K10} shows the effectiveness of collaborative learning in achieving near-optimal speedup. For each value $M$, we adjust the value of the communication threshold $D$ as shown in the table to achieve a near-optimal speedup. The sample complexity of the centralized algorithm for a single SBS is $T_{\mathcal{O}} = 64631.15$ samples, which represents the minimum number of samples required to achieve the desired confidence level. If $M$ SBSs learn independently, the total sample complexity would be $M T_{\mathcal{O}}$. Whereas the sample complexity of a centralized $\mathcal{XY}$-Oracle algorithm is $542901.40$ samples. It is higher than that of the DistLinGapE, despite knowing  $\boldsymbol{\theta}^*$, because the DistLinGapE has a tighter confidence bound, allowing for more accurate reward estimates. 

When the SBSs collaborate, the DistLinGapE achieves a near-optimal speedup $M$. Table~\ref{tab:speedup_K10} shows the value of the communication threshold $D$ and the number of communication rounds needed to achieve $\mathcal{S}_{\mathcal{A}}$ speedup in the simulation example with $M$ collaborating agents.
\renewcommand{\arraystretch}{1.5}
\begin{table}[ht] 
   \caption{Algorithm Speedup}
    \centering
    \begin{tabular}{|c|c|c|c|c|} \hline
      M              &    1       & 2          & 4         & 6      \\ \hline
      D              & $\infty$   &  10        & 1         & 0.1  \\ \hline
      samples/agent  & 64631.12   &  30052.02  &  17051.46   & 10914.82 \\ \hline
      $S_\mathcal{A}$&    -       &   1.86     &  3.79     & 5.92    \\ \hline
      Communication rounds& 0     &  86.40     & 128.93     & 342.65   \\ \hline
    \end{tabular}\label{tab:speedup_K10}
\end{table}

Fig.~\ref{fig:D_vs_comm_samples} illustrates the impact of the communication threshold, $D$, on the number of communication rounds performed by an agent and the total sample complexity required to achieve the desired confidence. The results correspond to a bandit instance with $K=10$ services and $M=4$ SBSs. Small values of $D$ lead to excessive communication overhead, as agents exchange information too frequently, without improving the sample complexity. Conversely, selecting a large $D$ increases the total sample complexity, as each agent performs a greater number of local updates before sharing information. This trade-off highlights the importance of carefully tuning $D$ to balance communication efficiency and sample complexity.
\begin{figure}[ht]
    \centering
    \includegraphics[width=0.8\linewidth]{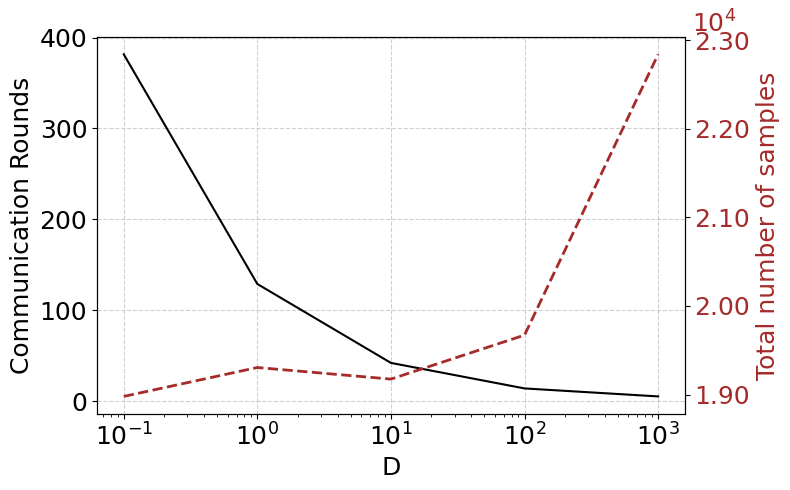}
    \caption{The communication threshold $D$ vs. the number of communication rounds and the total sample complexity for a four-SBS network with $K=10$ services.}
    \label{fig:D_vs_comm_samples}
\end{figure}

Fig.~\ref{fig:revision_Kvar} compares the sampling complexity and the number of communication rounds for a system with $M=4$,
$K \in \{10,20,30,40,50\}$, and 300 users per SBS. It shows that the sampling complexity is mainly affected by the gaps between the arms rather than by the number of arms itself. This explains why, in the figure, the sampling complexity at $K=50$ is significantly lower than that at $K=40$, for example. This behavior is consistent with the sampling complexity bound derived in
Theorem~\ref{th:sample_complexity}, where the dependence on $K$ appears only logarithmically, while the dominant factor is the set of arm gaps that determine the problem hardness $H_{\epsilon}$.  
We note that, in this example, we increased the number of users per SBS from 75 to 300. With 75 users, the aggregated delay improvement over all users was too small, resulting in very slow confidence shrinkage and very long sampling times.

\begin{figure}[ht]
    \centering
    \includegraphics[width=0.8\linewidth]{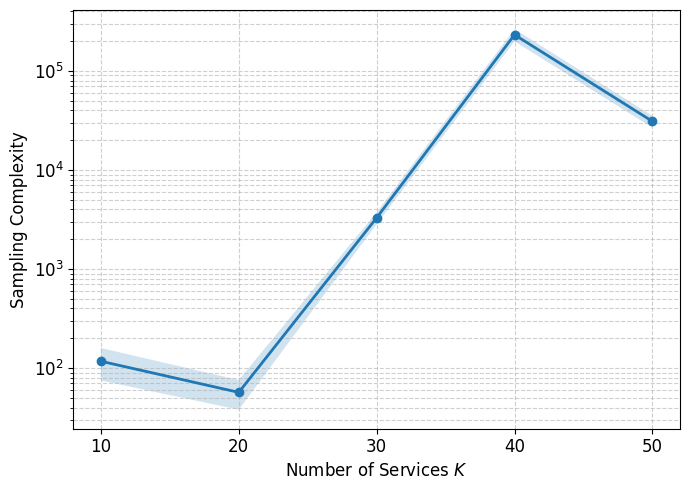}
    \caption{The sampling complexity and number of communication of the DistLinGapE for different number of services (arms).}
    \label{fig:revision_Kvar}
\end{figure}

Fig.~\ref{fig:delay_OFUL} and Fig.~\ref{fig:delay_OFUL_round} show the learning behavior of DistLinGapE and centralized batch OFUL. For a fair comparison, since DistLinGapE samples 
$M$ arms in each global round, a centralized 
$M$-batch OFUL is used to select the 
$M$ arms with the highest indices. Both algorithms are run for 400 rounds with $D=100$ for DistLinGapE. Fig.~\ref{fig:delay_OFUL} shows that, under the same number of rounds and parallel samples, DistLinGapE achieves a higher cumulative delay reduction, indicating more effective sampling. Fig.~\ref{fig:delay_OFUL_round} explains this behavior by presenting the average delay improvement per round. OFUL initially has a lower average delay improvement as it explores less certain services (least explored) based on its optimism sampling principle, whereas DistLinGapE concentrates sampling on services that are hardest to distinguish from the optimal one, enabling earlier sampling of delay efficient services.

\begin{figure}
    \centering
    \includegraphics[width=0.8\linewidth]{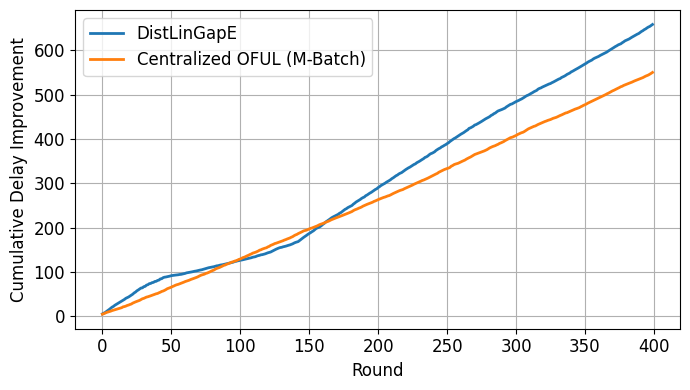}
    \caption{The cumulative delay improvement for the DistLinGapE with $M=4$ SBSs and centralized OFUL with $M$-batch learning. }
    \label{fig:delay_OFUL}
\end{figure}

\begin{figure}
    \centering
    \includegraphics[width=0.8\linewidth]{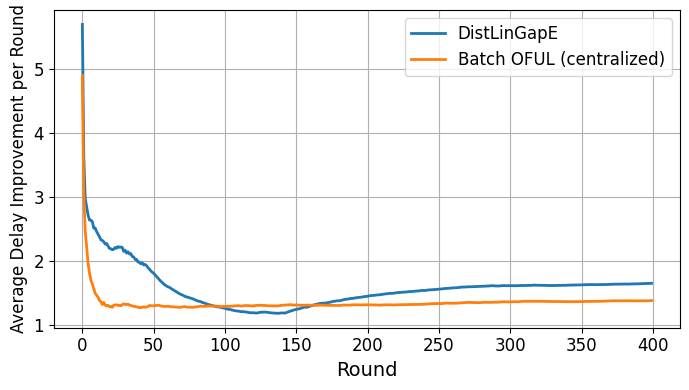}
    \caption{The average delay improvement per round for the DistLinGapE with $M=4$ SBSs and centralized OFUL with $M$-batch learning. }
    \label{fig:delay_OFUL_round}
\end{figure}

\section{Conclusion and Future Work}
\label{sec:conclusion}
This work addresses the problem of service placement in small cell networks, considering the unknown service demand and the randomness of the environment. Given the limited resources at the SBSs, the goal was to identify the service that would result in the highest reduction in the total user delay when deployed at the SBSs instead of the cloud. To achieve this, we modeled the service demand as a linear function of service attributes and formulated the service placement problem as a distributed BAI problem.  
We then proposed a novel distributed multi-agent BAI algorithm for the fixed-confidence setting under the linear bandit framework. Our numerical analysis demonstrated the efficacy of the proposed approach in finding the optimal arm (service) and that collaboration between agents (SBSs) enables near-optimal speedup, reducing the sample complexity per agent by a factor of $M$ compared to the case where the agents learn independently. Additionally, we derived upper bounds on both the sample complexity and the number of communication rounds required for the proposed algorithm.  
To the best of our knowledge, this work is the first to apply the BAI variant of MAB to the service placement problem. However, several directions remain for future research. One potential extension is to consider the top-$m$ arms setting, where the algorithm identifies the best $m$ arms instead of just one. This would allow service providers to deploy multiple services at SBSs, providing a more practical and flexible solution. Another promising direction is to extend this work to a heterogeneous settings with different contexts across agents. Furthermore, the overlap in coverage regions introduces coupling of requests, which must be addressed in future algorithms. Beyond algorithm design, theoretical guarantees, particularly upper bounds on sample complexity and communication rounds, need to be derived for these cases.

\appendix
\input{supplemental}

\input{main.bbl}
\begin{IEEEbiography}
[{\includegraphics[width=1in,height=1.25in,clip,keepaspectratio]{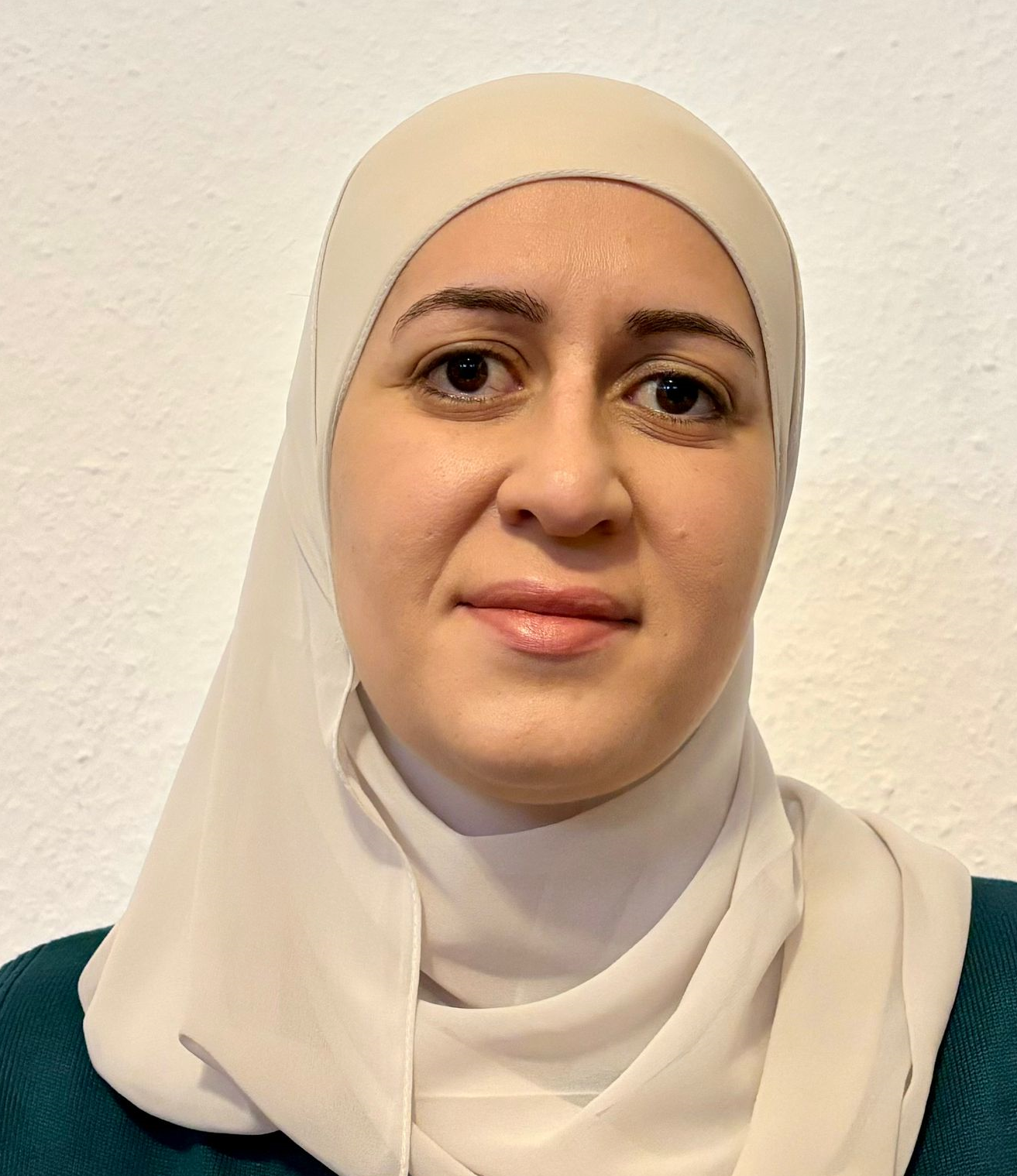}}]{Mariam Yahya} received the BSc degree in electrical engineering from Birzeit University, Palestine, the 
MSc degree in wireless communications from the Jordan University of Science and Technology, Jordan, and the PhD degree in computer science from  the University of Tübingen, Germany, in 2025. She was a lecturer with Palestine Technical University. She is currently a postdoctoral researcher with Ruhr University, Bochum, Germany. Her research interests  include wireless communications, network optimization, multi-armed bandits, and federated learning.
\end{IEEEbiography}

\begin{IEEEbiography}[{\includegraphics[width=1in,height=1.25in,clip,keepaspectratio]{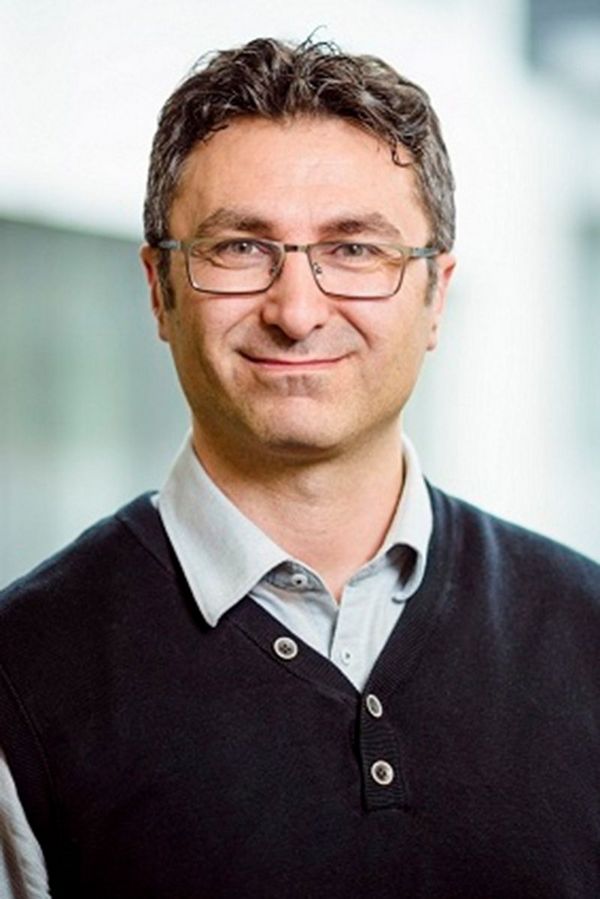}}]{Aydin Sezgin} (Senior Member, IEEE) received the Dr.-Ing degree in electrical engineering from TU Berlin in 2005. From 2006 to 2009, he was a postdoctoral with Information Systems Laboratory, Stanford University and Department of Electrical Engineering and Computer Science, University of California, Irvine. 
He is currently a professor with Ruhr University, Bochum, Germany. He has authored or coauthored 
several book chapters, more than 65 journals, and 200 conference papers.
\end{IEEEbiography}

\begin{IEEEbiography}[{\includegraphics[width=1in,height=1.25in,clip,keepaspectratio]{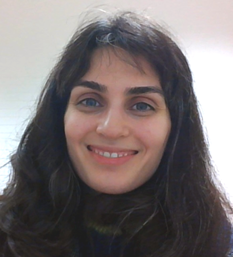}}]{Setareh Maghsudi} received the PhD degree (summa cum laude) from the Technical University of Berlin in 2015. From 2015 to 2017, she was a postdoctoral with the University of Manitoba, Canada, and Yale University, USA. From 2017 to 2023, she was an assistant professor with the Technical University of Berlin and University of Tübingen. She is currently a full professor with Ruhr-University, Bochum and senior researcher with Fraunhofer Heinrich-Hertz Institute, Berlin. Her research interests include the intersection of network analysis and optimization, game theory, machine learning, and data science. She was the recipient of several competitive fellowships, awards, and research grants from different institutes, including the 
German Research Foundation, German Ministry of Education and Research, and Japan Society for the Promotion of Science.
\end{IEEEbiography}


\end{document}

%% file: supplemental.tex
%
%




\appendices
\section{Auxiliary Lemmas}
In this section, we present the lemmas necessary for proving the theorems stated in the paper.
\begin{lemma}
\label{lem:good_event}
    Let $\Delta(i,j) =(\boldsymbol{x}_i - \boldsymbol{x}_j)^{\top} \boldsymbol{\theta}^{*}$ be the gap of the expected rewards between arms $i$ and $j$, and define a good event $\mathcal{E}_m$ for agent $m$ as follows \cite{xu2018fully}
    \begin{equation}
        \mathcal{E}_m = \left\{\forall t >0, \forall i,j \in [K], \lvert \Delta(i,j) - \hat{\Delta}_{t,m} (i,j) \rvert \le \beta_{t,m}(i,j) \right\}.
    \label{eq:event}
    \end{equation}
    The event $\mathcal{E}_m$ occurs with probability at least $1-M\delta_m$.
\end{lemma}
\begin{proof}
    This lemma results from \eqref{eq:reward_diff} in Proposition~\ref{prop:confidence_ellipsoid} followed by the union bound.
\end{proof}

The following lemma is based on \cite[Lemma 5]{xu2018fully}.
\begin{lemma}\label{lem:bound_B} 
    Under event $\mathcal{E}_m$ at agent $m$, $B_m(t)$ is bounded by 
    \begin{align}
        B_m(t) \le  & \min \left(0, -\max\left(\Delta_{i_{t,m}}, \Delta_{j_{t,m}} \right) +\beta_{t,m}(i_{t,m}, j_{t,m})  \right) \nonumber \\ 
        &+ \beta_{t,m} (i_{t,m}, j_{t,m}),
    \end{align}
    if either $i_{t,m}$ or $j_{t,m}$ is the best arm. Otherwise,
    \begin{align}
        B_m(t) \le & \min\left(0, -\max\left(\Delta_{i_{t,m}}, \Delta_{j_{t,m}} \right) + 2 \beta_{t,m}(i_{t,m}, j_{t,m})\right) \nonumber \\ &+ \beta_{t,m} (i_{t,m}, j_{t,m}).
    \end{align}
\end{lemma}
Lemma~\ref{lem:T_mt} follows directly from \cite[Lemma 2]{xu2018fully}, adapted to the multi-agent setting by accounting for the number of times an arm is pulled, including the cumulative pulls by other agents in previous rounds as reported by the central coordinator.
\begin{lemma}\label{lem:T_mt}
    Let $T_{m,i}(t)$ be the number of times  arm $i$ at agent $m$ is sampled before round $t$, including the samples by other agents as reported by the central coordinator in the most recent communication round. Then, the norm $\|\X_i - \X_j \|_{\boldsymbol{A}_{t,m}^{-1}}$ is bounded by
    \begin{equation}
        \|\X_i - \X_j \|_{\boldsymbol{A}_{t,m}^{-1}} \le \sqrt{\frac{\alpha(i,j)}{T_{m,i,j} (t)}},
    \end{equation}
    where 
    \begin{equation}
        T_{m,i,j}(t) = \underset{k \in [K]:p^{*}_k(i,j)>0}{\min}T_{m,k}(t)/ p^{*}_k (i,j).
    \end{equation}
\end{lemma}
%
\begin{lemma}
    \textbf{The Determinant-Trace Inequality \cite[Lemma 10]{abbasi2011improved}} Suppose that $\X_{a_1}, \X_{a_2}, \dots , \X_{a_t} \in \Rd$, for $1 \le s\le t$ and $\lVert \X_{a_s} \rVert_2 \le L$. Let $\bar{\boldsymbol{A}}_t = \lambda \boldsymbol{I} + \sum_{s=1}^t \X_{a_s} \X_{a_s}^{\top}$ for some $ \lambda > 0$. Then,
    \begin{equation}
        \det(\bar{\boldsymbol{A}}_t) \le (\lambda + t L^2/d)^d.
    \end{equation}
    \label{lemma:det_trace}
\end{lemma}

\section{Proof of Theorem~\ref{th:opt_arm}}
\label{app:confidence}
\begin{proof}
Let arm $\hat{a}_m^{*}$ be the arm returned by agent $m$ at the stopping time $T$. If this arm is worse than $ \hat{a}^{*}$ by at least $\epsilon$, $\Delta(a^{*}, \hat{a}_m^{*}) > \epsilon$, then, by the definition of the stopping condition $B_m(T) \le \epsilon$ we get
\begin{equation}
    \Delta(a^{*}, \hat{a}_m^{*}) > \epsilon \ge B_m(T) \ge \hat{\Delta}_{T,m}(a^{*}, \hat{a}_m^{*}) +\beta_{T,m} (a^{*}, \hat{a}_m^{*}).
\end{equation}
Using \eqref{eq:event} and Lemma~\ref{lem:good_event}, the probability of returning a suboptimal arm at one of the agents is
\begin{equation}
    \mathbb{P} [\Delta(a^{*}, \hat{a}_m^{*}) > \epsilon] \le \mathbb{P} [\bar{\mathcal{E}}_m] = 1- \mathbb{P} [\mathcal{E}_m] \le M \delta_m =\delta,
\end{equation}
where $\bar{\mathcal{E}}_m$ is the complement of the event $\mathcal{E}_m$. This proves the condition in \eqref{eq:stop_cond}.
\end{proof}
\section{Proof of Theorem~\ref{th:sample_complexity}}
\label{app:sample_complexity}
\begin{proof}
Let agent $m^{*}$ be the first agent to satisfy the stopping condition $B_{m^*}(t) \le \epsilon$. Also, let $\tilde{t}$ be the last round that arm $k$ is pulled at agent $m^*$. Then, from Lemma~\ref{lem:bound_B}, and because the stopping condition is not yet met, we have
\begin{align} \label{eq:bound_on_Bm}
\min &\left(  0,-\Delta_k+2\beta_{\tilde t-1,m^*}(i_{\tilde t-1,m^*},j_{\tilde t-1,m^*}) \right)
 \nonumber \\ 
    &+\beta_{\tilde t-1,m^*}(i_{\tilde t-1,m^*},j_{\tilde t-1,m^*}) 
     \geq B_m(\tilde t-1) \geq \epsilon.
\end{align}
In Lemma~\ref{lem:bound_B}, $T_{m^*,i,j}(t)$ represents the total number of times arms $i$ and $j$ have been selected together up until time $t$. Due to collaboration, this value is not limited to the number of pulls at agent $m^*$ but includes the pulls across all agents in past communication rounds. 

By substituting the value of $\beta_{\tilde t-1,m^*} =  \|\boldsymbol{x}_{i_{\tilde t-1,m^*}} - \boldsymbol{x}_{j_{\tilde t-1,m^*}} \|_{\boldsymbol{A}_{\tilde t-1,m^*}}^{-1} C_{\tilde t-1,m^*} $ in  \eqref{eq:bound_on_Bm} and then using Lemma~\ref{lem:T_mt} we obtain
\begin{equation}
  T_{m^*,i_{\tilde t-1},{j_{\tilde t-1}}} \leq \frac{\alpha(i_{\tilde t-1,m^*},j_{\tilde t-1,m^*})}{\max\left(\epsilon, \frac{\epsilon+\Delta_{i_{\tilde t-1,m^*}}}{3},\frac{\epsilon+\Delta_{j_{\tilde t-1,m^*}}}{3}\right)^2}
  C^2_{\tilde t-1,m^*},
\end{equation}
where $C_{\tilde t-1,m^*}$ is defined in Proposition~\ref{prop:confidence_ellipsoid}.  

When arm $k$ is pulled at $\tilde t$-th round at agent $m^*$, 
\begin{equation}
T_{m^*,k}(\tilde t-1) = p^*_k(i_{\tilde t-1,m^*},j_{\tilde t-1,m^*}) T_{m^*,i_{\tilde t-1},{j_{\tilde t-1}}}(\tilde{t} -1), 
\end{equation} 
holds by definition. Therefore, the number of times arm $k$ is pulled by the final round $T$, $T_{m^*,k}(T)$, is bounded by
\begin{align}
&T_{m^*,k}(T) =  T_{m^*,k}(\tilde t-1) + 1 \nonumber \\
&=  p^*_k(i_{\tilde t-1,m^*},j_{\tilde t-1,m^*}) T_{m^*,i_{\tilde t-1}, j_{\tilde t-1}} (\tilde t-1) + 1 \nonumber\\
&\leq \max_{i,j\in [K]}  p^*_k(i,j)  T_{m^*,i_{\tilde t-1}, j_{\tilde t-1,m^*}}(\tilde t-1) + 1 \nonumber \\
&\leq \frac{p^*_k(i_{\tilde t-1,m^*},j_{\tilde t-1,m^*})\alpha(i_{\tilde t-1,m^*},j_{\tilde t-1,m^*})}
{  \max\left(\epsilon, \frac{\epsilon+\Delta_{i_{\tilde t-1}}}{3},\frac{\epsilon+\Delta_{j_{\tilde t-1}}}{3}\right)^2}
C^2_{\tilde t-1,m^*} + 1 \nonumber \\
&\leq \max_{i,j\in [K]} \frac{p^*_k(i,j)\alpha(i,j)}{\max\left(\epsilon, \frac{\epsilon+\Delta_i}{3},\frac{\epsilon+\Delta_j}{3}\right)^2}C^2_{T,m^*} + 1. 
\end{align}
The total sample complexity across all agents is the sum of the number of samples across all arms, we have $\tau = \sum_{k=1}^K T_{m^*,k}(T)$, which yields 
\begin{align}
    \tau \leq H_\epsilon C^2_{T,m^*} + K,
    \label{eq:tau_total}
\end{align}
To find the per-agent sample complexity $\tau_m$ we note that $\tau_m =\lceil \tau/M \rceil$, so $\tau_m \le \tau/M+1$. Using the upper bound in \eqref{eq:tau_total}, we have
\begin{equation}
    \tau_m \leq \Hm C^2_{T,m^*} + \mu,
    \label{eq:bound_tau_m}
\end{equation}
where $\mu = \frac{K}{M} +1$. Since the agents select arms sequentially, the sample complexity per agent is equal for all agents. Now, we find the upper bound on $C^2_{T,m^*}$. 
\\
\textbf{Bounding the Confidence Term:}

We extend the proof in \cite{xu2018fully} to the multi-agent case. From Proposition~\ref{prop:confidence_ellipsoid} we have 
\begin{equation}
    C_{T,m^*} = R\sqrt{2\log\frac{K^2\det(\boldsymbol{A}_{T,m^*})^{\frac12}\det(\lambda \boldsymbol{I})^{-\frac12}}{\delta_m}} + \lambda^{\frac12}S  .
\end{equation}
Using Lemma~\ref{lemma:det_trace}, and the fact that by the end of the algorithm at time $T$, the number of arm pulls is $\tau$, we get 
\begin{equation}
    C_{T,m^*} \le R\sqrt{2\log\frac{K^2}{\delta_m} + d \log \left(1+ \frac{M \tau_m L^2}{\lambda d} \right) } + \lambda^{\frac12}S ,
    \label{eq:conf_M_tau_m}
\end{equation}
where we used the upper limit on the number of samples $\tau \le M \tau_m$. To find an upper bound on the expression in \eqref{eq:conf_M_tau_m}, we consider two cases for $\lambda$:
\subsubsection{\textbf{Case 1}}  For $\lambda \leq \frac{2R^2}{S^2}\log\frac{K^2}{\delta_m}$.
In this case, \eqref{eq:conf_M_tau_m} can be expressed as
\begin{equation}
        C_{T,m^*} \le  2R\sqrt{2\log\frac{K^2}{\delta_m} + d\log\left(1+\frac{M \tau_m L^2}{\lambda d}\right)}.
\end{equation}
Substituting this value in \eqref{eq:bound_tau_m},
\begin{align}
    \tau_m &\leq \Hm C^2_{T,m^*} +\mu  \nonumber \\
    &\leq 4  \Hm R^2 \left(2\log\frac{K^2}{\delta_m} + d\log\left(1+\frac{M\tau_m L^2}{\lambda d}\right)\right) + \mu.
\end{align}
Let $\tau'_m$ be the parameter satisfying
\begin{equation}
    \tau_m = 4  \Hm R^2 \left(2\log\frac{K^2}{\delta_m} + d\log\left(1+\frac{M\tau'_m L^2}{\lambda d}\right)\right) + \mu,\label{eq:tau-bound}
\end{equation}
then, $\tau'_m \leq \tau_m$ holds. 

Let $N = 8 \Hm R^2\log\frac{K^2}{\delta_m} + \mu$, we have
\begin{align}
    \tau'_m \leq \tau_m & = 4 \Hm R^2d\log\left(1+\frac{M \tau'_m L^2}{\lambda d}\right)+ N \nonumber  \\
    & \leq 4  \Hm R^2 \sqrt{dL^2 M \tau'_m/\lambda} + N.
\end{align}
By solving this inequality for $\tau_m'$ we obtain
\begin{equation}
    \sqrt{\tau'_m} \leq 2\sqrt{16 H_\epsilon^2 R^4dL^2/(M\lambda) +N^2}.
\end{equation}
Let $Y = 2\sqrt{16  H_\epsilon^2 R^4dL^2/(M\lambda) +N^2}$, then, using the upper bound of $\tau'_m$ in \eqref{eq:tau-bound} yields 
\begin{equation}
    \tau_m \leq \mu + 4 \Hm R^2 \left(
    2   \log\frac{K^2}{\delta_m}  +
      d\log\left(1+\frac{Y^2 L^2}{\lambda d}\right)
    \right) .
\end{equation}
\subsubsection{\textbf{Case 2}}: For $\lambda > 4  H_\epsilon R^2L^2$:

Using the inequality $(a+b)^2 \le 2(a^2 + b^2)$, we can rewrite \eqref{eq:conf_M_tau_m} as
\begin{equation}
    C_{T,m^*} \le 2\left(2 R^2 \log\frac{K^2}{\delta_m} + \frac{M \tau_m R^2 L^2}{\lambda} + \lambda S^2 \right).
\end{equation}
Substituting this value in \eqref{eq:bound_tau_m} yields, 
\begin{align}
    \tau_m &\le  \Hm C_{T,m^*}+ K \notag \\
    & \le 2 \Hm \left(2 R^2 \log\frac{K^2}{\delta_m} + \frac{M \tau_m R^2 L^2 }{\lambda} + \lambda S^2 \right) + \mu.
\end{align}
Solving for $\tau_m$, we get
\begin{align}
    \tau_m &\le  \left( 1- \frac{2  H_\epsilon R^2 L^2}{ \lambda}\right)^{-1} \left(\frac{4 H_\epsilon R^2}{M} \log \frac{K^2}{\delta_m} +   \frac{2 H_\epsilon \lambda S^2}{M}  +\mu \right) \notag \\
     & \le 2 \left( \frac{4 H_\epsilon R^2}{M} \log\frac{ K^2}{\delta_m} +  \frac{2 H_\epsilon \lambda S^2}{M} + \mu \right).
\end{align}
\end{proof}
\section{Proof of Theorem~\ref{th:comm_cost}} 
\label{app:comm_cost}
%
\begin{proof}
In DistLinGapE, agents operate in parallel and are synchronized, meaning that they communicate with the central coordinator at the same time instants. We divide the running time into epochs, where each epoch ends with one communication round. Epoch $p$ starts immediately after the $p$-th synchronization event and ends at the $(p+1)$ synchronization event. Let $P$ denote the total number of epochs until stopping.

Let $T$ be the stopping time measured in global rounds. Let $\tau_m$ denote
the total number of local samples collected by agent $m$ up to stopping, and
let $\tau \coloneqq \sum_{m=1}^M \tau_m$ be the total sampling complexity
across all agents. For each epoch $p$ and agent $m$, let $\Delta \tau_{p,m}$ be
the number of local samples collected by agent $m$ during epoch $p$, so that
\begin{equation*}
\sum_{p=0}^{P-1} \Delta \tau_{p,m} = \tau_m,
\qquad
\sum_{p=0}^{P-1} \sum_{m=1}^M \Delta \tau_{p,m} = \tau.
\end{equation*}

Let $\boldsymbol{A}_p$ denote the aggregated regularized information matrix at the start of
epoch $p$, that is,
$
\boldsymbol{A}_p
=
\lambda \boldsymbol{I}
+
\sum_{\ell=0}^{p-1}
\sum_{m=1}^M
\Delta \boldsymbol{A}_{\ell,m}$, and let $\Delta \boldsymbol{A}_{p,m}$ be the local information matrix
accumulated by agent $m$ during epoch $p$. Hence, the aggregated matrix at the end of epoch $p$ satisfies
\begin{equation*}
    \boldsymbol{A}_{p+1} = \lambda \boldsymbol{I} +\boldsymbol{A}_p + \sum_{m=1}^M \Delta \boldsymbol{A}_{p,m}.
\end{equation*}

Fix $\gamma>0$ and define an epoch $p$ to be long if
$\max_{m\in[M]} \Delta \tau_{p,m} \ge \gamma$, and short otherwise.
Since each long epoch satisfies $\max_m \Delta\tau_{p,m}\ge \gamma$, we have the number of long epochs $N_{\text{long}}$ bounded by
\begin{equation*}
N_{\text{long}} 
\le
\frac{1}{\gamma}\sum_{p=0}^{P-1}\max_{m}\Delta\tau_{p,m}
\le
\frac{1}{\gamma}\sum_{p=0}^{P-1}\sum_{m=1}^M \Delta\tau_{p,m}
=
\frac{\tau}{\gamma}.
\end{equation*}

Now consider a short epoch $p$, for which $\Delta\tau_{p,m}<\gamma$ for all $m$.
Communication at the end of epoch $p$ is triggered by at least one agent; thus
there exists an agent $m_p$ such that
\begin{equation}
\Delta\tau_{p,m_p}\;
\log\!\frac{\det\!\big(\boldsymbol{A}_p+\Delta \boldsymbol{A}_{p,m_p}\big)}{\det(\boldsymbol{A}_p)}
\;>\; D.
\label{eq:trigger_short_epoch_nosub}
\end{equation}
Since $\Delta\tau_{p,m_p}<\gamma$, \eqref{eq:trigger_short_epoch_nosub} implies
\begin{equation}
\log\!\frac{\det\!\big(\boldsymbol{A}_p+\Delta \boldsymbol{A}_{p,m_p}\big)}{\det(\boldsymbol{A}_p)}
\;>\; \frac{D}{\Delta\tau_{p,m_p}}
\;\ge\; \frac{D}{\gamma}.
\label{eq:short_epoch_gain_local_nosub}
\end{equation}
Moreover, since $\sum_{m=1}^M \Delta \boldsymbol{A}_{p,m}\succeq \Delta \boldsymbol{A}_{p,m_p}$ and
$\boldsymbol{A}_p\succ \boldsymbol{0}$, PSD monotonicity yields
\begin{equation*}
\det(\boldsymbol{A}_{p+1})
=
\det\!\Big(\boldsymbol{A}_p+\sum_{m=1}^M \Delta \boldsymbol{A}_{p,m}\Big)
\;\ge\;
\det\!\big(\boldsymbol{A}_p+\Delta \boldsymbol{A}_{p,m_p}\big),
\end{equation*}
and hence
\begin{equation}
\log\!\frac{\det(\boldsymbol{A}_{p+1})}{\det(\boldsymbol{A}_p)}
\;\ge\;
\log\!\frac{\det\!\big(\boldsymbol{A}_p+\Delta \boldsymbol{A}_{p,m_p}\big)}{\det(\boldsymbol{A}_p)}
\;>\; \frac{D}{\gamma}.
\label{eq:short_epoch_gain_global_nosub}
\end{equation}
Therefore, each short epoch contributes at least $D/\gamma$ to the cumulative
log-determinant growth, and number of short epochs $N_{\text{short}}$
\begin{equation*}
N_{\text{short}} \cdot \frac{D}{\gamma}
<
\sum_{p=0}^{P-1} \log\!\frac{\det(\boldsymbol{A}_{p+1})}{\det(\boldsymbol{A}_p)}
=
\log\!\frac{\det(\boldsymbol{A}_P)}{\det(\boldsymbol{A}_0)}.
\end{equation*}

To bound the right hand side, note that $\boldsymbol{A}_P$ equals the aggregated information
matrix at stopping and includes $\tau$ total rank-one updates across all agents.
By the determinant--trace inequality in Lemma~\ref{lemma:det_trace}, we obtain
\begin{equation*}
\det(\boldsymbol{A}_P)
\le
\lambda^d\Big(1+\frac{\tau}{\lambda d}\Big)^d,
\end{equation*}
and since $\boldsymbol{A}_0=\lambda \boldsymbol{I}$,
\begin{equation*}
\log\!\frac{\det(\boldsymbol{A}_P)}{\det(\boldsymbol{A}_0)}
=
\log\!\frac{\det(\boldsymbol{A}_P)}{\lambda^d}
\le
d\log\!\Big(1+\frac{\tau}{\lambda d}\Big)
=
O\!\big(d\log \tau\big),
\end{equation*}
where $\lambda$ and $d$ are fixed problem parameters. Hence,
\begin{equation*}
N_{\text{short}}
\le
\frac{\gamma}{D}\; d\log\!\Big(1+\frac{\tau}{\lambda d}\Big).
\end{equation*}

We have $P=N_{\mathrm{long}}+N_{\mathrm{short}}$, so
\begin{equation*}
P
\le
\frac{\tau}{\gamma}
+
\frac{\gamma}{D}\; d\log\!\Big(1+\frac{\tau}{\lambda d}\Big).
\end{equation*}
Minimizing the right hand side over $\gamma>0$ yields
\begin{equation}
P
=
O\!\left(
\sqrt{
\frac{\tau\, d\log\!\big(1+\frac{\tau}{\lambda d}\big)}{D}
}
\right)
=
O\!\left(
\sqrt{
\frac{\tau\, d\log \tau}{D}
}
\right).
\label{eq:P_final_nosub}
\end{equation}

Finally, each synchronization event consists
of $M$ SBSs to coordinator transmission and $M$ coordinator to SBSs transmissions. Thus, the total
communication cost is $C_{\mathrm{tot}}=2MP$. Substituting \eqref{eq:P_final_nosub}
gives
\begin{equation*}
C_{\mathrm{tot}}
=
O\!\left(
M\sqrt{
\frac{\tau\, d\log \tau}{D}
}
\right).
\end{equation*}
\end{proof}
